\documentclass[showpacs,amsmath,amsfonts,amssymb,twocolumn,aps,pra,10pt,superscriptaddress,notitlepage,floatfix,nofootinbib,longbibliography]{revtex4-2}

\usepackage{graphicx} %
\usepackage{tikz}
\usetikzlibrary{quantikz2}
\usetikzlibrary{arrows}
\usetikzlibrary{decorations.pathmorphing,patterns}
\usetikzlibrary{calc}
\usepackage{amscd}
\usepackage{enumitem}
\usepackage{epsfig}
\usepackage{subfigure}
\usepackage{xcolor}
\usepackage{amsthm}
\usepackage{physics}
\usepackage{booktabs}
\usepackage{epstopdf}
\usepackage[most]{tcolorbox}
\usepackage{bbm}
\usepackage{framed}
\usepackage[ruled,vlined]{algorithm2e}
\usepackage{mathtools}
\usepackage{newtxtext}
\usepackage{svg}
\usepackage{silence}
\ErrorFilter{latex}{Unicode character}
\usepackage{csquotes}

\definecolor{googleblue}{RGB}{34, 0, 204}
\definecolor{panblue}{RGB}{0,24,150}
\definecolor{carmine}{RGB}{150, 0, 24}
\usepackage{hyperref} 
\hypersetup{
unicode=true,
bookmarksnumbered=false,
bookmarksopen=false,
breaklinks=false,
colorlinks=true,
linkcolor=carmine,
citecolor=googleblue,
urlcolor=panblue,
anchorcolor=OliveGreen}

\usepackage{diagbox}
\usepackage{makecell}
\usepackage[normalem]{ulem}

\usepackage[babel]{microtype} 
\microtypecontext{spacing=nonfrench}
\microtypesetup{
expansion={true,nocompatibility},
protrusion={true,nocompatibility},
activate={true,nocompatibility},
tracking=false,
kerning=true,
spacing={true}
}

\usepackage[all=normal,floats=tight,mathspacing=tight,wordspacing=tight,paragraphs=normal,tracking=tight,charwidths=tight,mathdisplays=normal,sections=normal,margins=normal]{savetrees}

\usepackage{newtxtext}

\graphicspath{{./Figures/}}

\newtheorem{theorem}{Theorem}
\newtheorem{lemma}{Lemma}
\newtheorem{proposition}{Proposition}
\newtheorem{corollary}{Corollary}

\newtheorem{definition}{Definition}

\newtheorem{remark}{Remark}

\newcommand{\elie}[1]{{\color{purple} #1}}

\newcommand{\GPT}{{\mathcal{GPT}}}
\newcommand{\IND}{\mathcal{I}}
\newcommand{\C}{\mathcal{C}}

\newcommand{\NEST}{{\mathcal{N\!M}}}

\newcommand{\BW}{{\mathcal{BW}}}
\newcommand{\MI}{{\mathcal{MI}}}

\newcommand{\MINF}{{\mathcal{MI{+}NF}}}
\newcommand{\VVS}{{\mathcal{VVS}}}
\newcommand{\MaxInt}{\operatorname{\textnormal{\textsf{MaxInt}}}}
\newcommand{\VVSplit}{\operatorname{\textnormal{\textsf{VVSplit}}}}
\newcommand{\Deflate}{\operatorname{\textnormal{\textsf{Deflate}}}}
\newcommand{\Proj}{\operatorname{\textnormal{\textsf{Proj}}}}
\newcommand\independent{\protect\mathpalette{\protect\independenT}{\perp}}
\def\independenT#1#2{\mathrel{\rlap{$#1#2$}\mkern2mu{#1#2}}}

\newtcolorbox[auto counter]{mybox}[2][]{
	enhanced,
	breakable,
	colback=blue!5!white,
	colframe=blue!75!black,
	fonttitle=\bfseries,
	title=Box \thetcbcounter: #2,#1
}

\makeatletter

\setbox0\hbox{$\xdef\scriptratio{\strip@pt\dimexpr
    \numexpr(\sf@size*65536)/\f@size sp}$}

\newcommand{\too}[1][3pt]{\mathrel{%
    \vcenter{\hbox{\rule[-.5\fontdimen8\scriptfont3]
               {\scriptratio\dimexpr#1\relax}{\fontdimen8\scriptfont3}}}%
   \mkern-4mu\hbox{\let\f@size\sf@size\usefont{U}{lasy}{m}{n}\symbol{41}}}}

\makeatother

\begin{document}
\preprint{APS/123-QED}
\title{On the physics of nested Markov models: a generalized probabilistic theory perspective}
\author{Xingjian Zhang}
\email{zxj24@hku.hk}
\affiliation{QICI Quantum Information and Computation Initiative, School of Computing and Data Science, The University of Hong Kong, Hong Kong SAR, China}
\author{Yuhao Wang}
\email{yuhaow@tsinghua.edu.cn}
\affiliation{Institute for Interdisciplinary Information Sciences, Tsinghua University, Beijing, China}
\affiliation{Shanghai Qi Zhi Institute, Shanghai, China}
\author{Elie Wolfe}
\email{ewolfe@perimeterinstitute.ca}
\affiliation{Perimeter Institute for Theoretical Physics, Waterloo, Ontario, Canada}

\begin{abstract}
Determining potential probability distributions with a given causal graph is vital for causality studies. To bypass the difficulty in characterizing latent variables in a Bayesian network, the nested Markov model provides an elegant algebraic approach by listing exactly all the equality constraints on the observed variables. However, this algebraically motivated causal model comprises distributions outside Bayesian networks, and its physical interpretation remains vague. In this work, we inspect the nested Markov model through the lens of generalized probabilistic theory, an axiomatic framework to describe general physical theories. We prove that all the equality constraints defining the nested Markov model are valid theory-independently. At the same time, not every distribution within the nested Markov model is implementable, not even via exotic physical theories associated with generalized probability theories (GPTs).  
To interpret the origin of such a gap, we study three causal models standing between the nested Markov model and the set of all distributions admitting some GPT realization. Each of the successive three models gives a strictly tighter characterization of the physically implementable distribution set; that is, each successive model manifests new types of GPT-inviolable constraints. We further demonstrate each gap through a specially chosen illustrative causal structure. We anticipate our results will enlighten further explorations on the unification of algebraic and physical perspectives of causality.
\end{abstract}
\maketitle

\section{Introduction}

Causal inference is a subfield of data science concerned with assessing causal relationships among statistically correlated observable variables~\citep{pearl2009causality,spirtes2000causation,rubin2005causal}. Given multivariate data, causal discovery is the task of determining the true underlying causal process that most likely generated the given data~\citep{spirtes2000causation}. 
To represent the structure of a causal process, one common choice is the directed acyclic graph (DAG), where each variable is denoted as a vertex and each direct causal influence is denoted as an edge~\citep{pearl2009causality,spirtes2000causation}. Given a certain number of observed variables in the statistical data, a candidate causal process to underlay that data must be a DAG with \emph{at least} as many vertices as there are variables. Any extra vertices that do \emph{not} correspond to an observed variable are to be understood as \emph{latent systems} and encode the possibility of \emph{indirect} causal relationships, namely unobserved common causes.

Existing approaches to discovering the underlying causal structure are commonly split into two categories: score-based approaches~\citep{chickering2002optimal,Solus2021} and constraint-based approaches~\citep{spirtes2000causation,spirtes2001anytime,colombo2014order}. While score-based causal discovery aims to provide relative likelihoods for different plausible causal hypotheses, \emph{constraint-based} causal discovery is concerned with simply partitioning the set of all candidate causal hypotheses into two sets: those causal structures for which the observed data are consistent with the operational constraints implied by the causal structure, and those for which the observed data evidently \emph{violate} an implied constraint.

In order to perform constraint-based causal discovery, we must first understand how to tease out the operational constraints implied by a given DAG. An operational constraint is a test that can be applied to statistical data. A constraint is said to be implied by a DAG if every distribution \emph{that the DAG can generate} satisfies the given constraint. The most prominent form of constraint considered in constraint-based causal discovery is known as a \emph{conditional independence relation}. A constraint of the form \enquote{$A$ is conditionally independent of $B$ given $C$} would require all distributions involving variables $\{A, B, C\}$ (potentially among any number of other variables) to satisfy
\begin{equation}
    p_{ABC}(a,b,c)=p_{A|C}(a|c)p_{B|C}(b|c)p_{C}(c)
\end{equation}
for all values $a$, $b$, and $c$. In the following, we abbreviate $p_{ABC}(a,b,c)$ as $p(a,b,c)$, and similarly for the others when it is clear from the context.

To understand if a given constraint (conditional independence relation or any other) is implied by a DAG, we need to better understand what it means for the DAG to \enquote{generate} a distribution. The conventional attitude is to say that a given DAG is capable of generating a given distribution if and only if there exists some \emph{structural equation model} corresponding to the DAG, which would generate a distribution over variables corresponding to \emph{all} the vertices in the DAG, such that the distribution over the \emph{observed} variables is obtained by \emph{marginalizing out} (i.e., summing or integrating over) the \emph{unobserved} variables~\citep{pearl2009causality}. In short: one treats latent systems as ordinary random variables that---for whatever reason---are not observed. As a corollary, if a distribution can be generated by a DAG, then it must fall within a \emph{Markov} model for this DAG. That is, the random variables $A, B$ in this distribution are conditionally independent given $C$ whenever their corresponding vertices satisfy the so-called ``$d$-separation criterion'' in the DAG~\citep{henson2014theory,pearl2009causality}.

But herein lies the rub. The conventional approach to latent systems fails when confronted with the task of accounting for the predictions of quantum theory. Modern physics has taught us that some correlations mediated by a \emph{quantum} common cause resist any classical explanation that attempts to preserve the same causal structure while restricting to conventional modelling of latent systems~\cite{wood2015lesson,cavalcanti2018classical,fritz2012beyond,henson2014theory,Chaves2015,tavakoli2022reviewnetworks}. Take the famous Clauser-Horne-Shimony-Holt (CHSH) Bell test as an example~\cite{bell1964einstein,clauser1969proposed}, whose causal structure is depicted in Figure~\ref{fig:Bell}(a). There, the violation of Bell inequalities indicates that the correlations among the observed variables cannot be explained by applying the conventional approach to the test's causal structure. This is a consequence of the quantum common cause, i.e., $\Lambda$.
Indeed, quantum theory and classical theory (i.e., the conventional approach) constitute two distinct possible frameworks for interpreting a given DAG with latent systems \cite{henson2014theory,spirtes2000causation}. 
The set of distributions over the observed variables consistent with a \emph{classical} interpretation is contained within the set of distributions consistent with the quantum interpretation \cite{brunner2014bell}. The failure of classical explanations to account for quantum predictions, then, occurs whenever that containment relationship is \emph{strict}. 
In summary, the set of constraints implied by a DAG is therefore a function not \emph{merely} of the DAG itself but rather of the DAG \emph{together with a physical theory for interpreting its latent systems}.

\begin{figure}[hbt!]
\centering 
\includegraphics[width=\columnwidth]{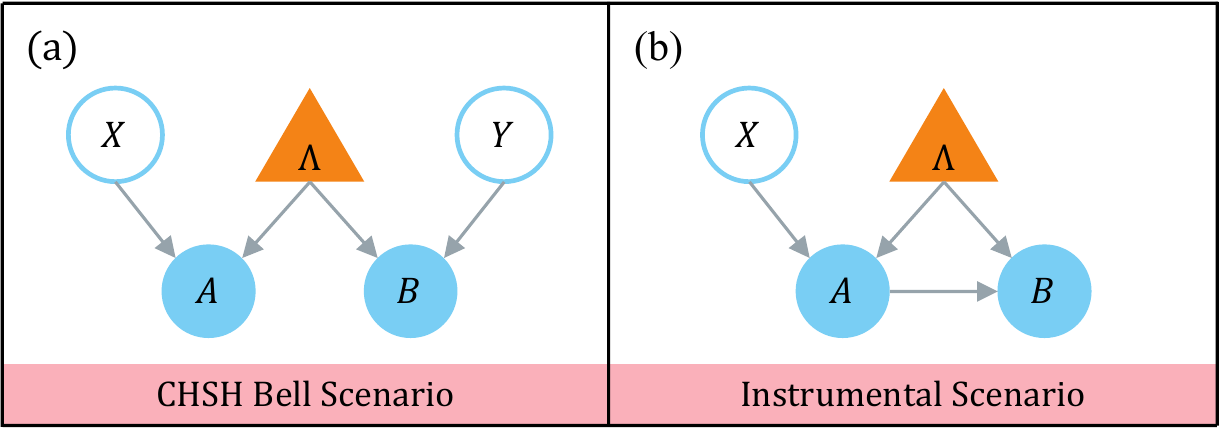}
\caption{DAGs for the CHSH Bell test and the instrument test. In a causal DAG, vertices represent variables, and directed edges represent causal influences, where the parental vertex represents the cause and the child vertex represents the effect. We distinguish observable variables from latent ones, as denoted by circles and triangles, respectively. 
(a) DAG of the CHSH Bell test. Vertices $X$ and $Y$ represent the measurement choices of the observers, $A$ and $B$ represent the measurement outcomes, and $\Lambda$ represents the latent variable, which is common cause of $A$ and $B$. In the CHSH test, all the observable variables take binary values. 
(b) DAG of the instrumental test. The vertex $X$ represents the instrumental variable, $A$ and $B$ are considered to be the direct cause and effect, and $\Lambda$ is the latent variable, which is common cause of $A$ and $B$.}
\label{fig:Bell}
\end{figure}

Upon appreciating this simple fact, a great deal of conventional wisdom is suddenly thrown into question: Given some set of constraints that hold \emph{classically} for a given DAG, do they also hold under the quantum interpretation? What sort of constraints are simultaneously well-justified by the classical interpretation yet violable by distributions admissible relative to the quantum interpretation? Could switching interpretation break the fundamental principle that $d$-separation relations in the DAG imply conditional independence relations in the DAG's compatible distributions? What about exotic physical theories that reduce to neither quantum nor classical physics?

In this work, we rebuild (from scratch) a hierarchical framework for deriving and justifying constraints that are \emph{inviolable across all physical theories}. We start by quoting a seminal result of~\citet{henson2014theory} that restores the primacy of $d$-separation. Namely, the set of distributions compatible with a DAG relative to \emph{any} physical interpretation must be contained in that DAG's Markov model.\footnote{A distribution is said to be in the Markov model for a given DAG if the distribution satisfies every conditional independence relation, which corresponds to a $d$-separation in the DAG over its observed variables.} We then claim (and eventually prove) a critical generalization, namely that the same containment relation even applies if we replace ordinary Markov models with nested Markov models \cite{tian2002testable,evans2017margins,richardson2023nested}. Roughly speaking, the nested Markov model is defined by exactly all the equality constraints governing the marginal model produced from the structural equations, comprising not only the conditional independencies but also the ones cited as the Verma constraint \cite{robins1986new,verma1990equivalence}.

The set of distributions compatible with a given DAG relative to \emph{any} physical theory is a concept central to everything that follows. To define this, we employ the framework of generalized probabilistic theories (GPTs)~\cite{barrett2007information,chiribella2010probabilistic,janotta2014generalized}. Each GPT is its own prescription for interpreting the distributions compatible with a DAG. In order to discuss the union of distributions across all possible GPTs we fallback on the even more general framework of operational probabilistic theories (OPTs)~\cite{chiribella2010probabilistic,dariano2017quantum}, following the same as-broad-as-possible approach taken by~\citet{henson2014theory}. A distribution admits \emph{some} GPT interpretation relative to a DAG only if it is consistent with \emph{the} OPT interpretation.

Thus, our nominal goal is to prove that OPT models are contained in the set of nested Markov models. We accomplish this proof by first defining a new model set---referred to as the \emph{Visible Vertex Split} model---which lies between the nested Markov model and OPTs: it implies the satisfaction of all nested Markov constraints while at the same time transparently containing all OPT-compatible distributions. In contrast to the nested Markov models, the new model manifests a type of \emph{inequality} constraints that are inviolable by any physical theory. These inequalities are consistent with a graphical criterion called the $e$-separation~\cite{evans2012graphical,finkelstein2021entropic}, which captures the idea of causal bottlenecks even between sets of variables not separated by any conditional independence relation.  Still, the \emph{Visible Vertex Split} model comprises distributions that are outside OPTs and hence non-physical. We then explore this gap by studying two further models, which we call the \emph{Maximal Interruption} model and a model further enhanced by \emph{Nonfanout Inflation} \cite{wolfe2021quantum}. As shown in Figure~\ref{fig:Venn}, the three new models all include GPTs and tighten the distribution set successively. Each model provides a new type of inequality constraints, which can be induced by a particular graphical criterion. We leave the question of whether further refinement might be possible relative to GPTs instead of the OPT framework as a desideratum for future work.

\begin{figure}[hbt!]
\centering 
\includegraphics[width=\columnwidth]{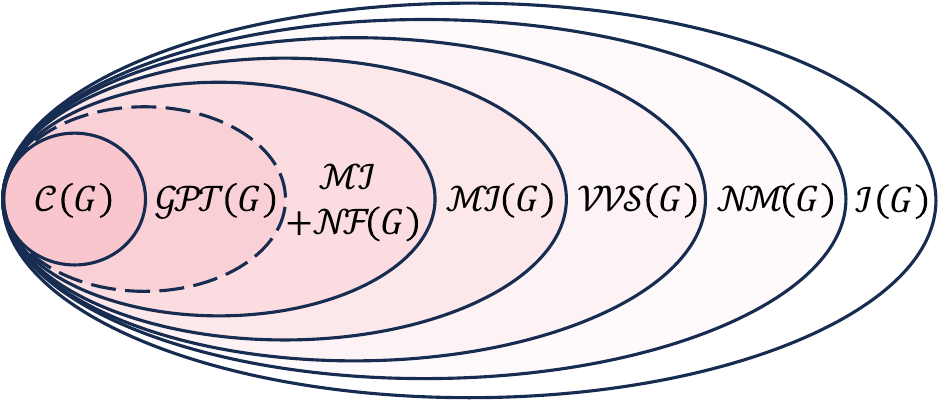}
\caption{Venn diagram of the causal models investigated in this work. In the literature, it is known that the classical model, $\C(G)$, is included in the GPT model, $\GPT(G)$, and the latter is further included in the causal model defined by all the conditional independencies, $\IND(G)$. In causal structures like the Bell test, each inclusion can be made strict. In this work, we further investigate the nested Markov model, $\NEST(G)$, a causal model that was known to provide a strictly tighter characterization than $\IND(G)$. We prove that $\GPT(G)\subseteq\NEST(G)$, and we show that there exist causal structures that make the inclusion strict. In the study of the nested Markov model, we further propose the visible-vertex-split model, $\VVS(G)$, the maximal interruption model, $\MI(G)$, and the maximal interruption plus nonfanout inflation model, $\MINF(G)$. We show $\MINF(G)\subseteq\MI(G)\subseteq\VVS(G)\subseteq\NEST(G)$, and there exist causal models that make each inclusion strict. Following \citet{wolfe2021quantum}, it is known $\GPT(G)\subseteq\MINF(G)$, but it remains open whether this inclusion is strict for certain causal structures.}
\label{fig:Venn}
\end{figure}

\section{Results}
\subsection{Basic concepts}
Before we commence, we first review the basic concepts of GPTs and causality. Starting from the investigation into the interpretation of quantum physics, the GPT framework was proposed to study and compare various physical theories~\cite{barrett2007information}. As the name suggests, an underlying assumption of GPTs is the validity of probabilities. In statistical theory, we can abstract the essential features of probabilities and formalize the concept of a probability kernel. Following~\citet{lauritzengraphical}, given finite state spaces, $\mathcal{X}_A,\mathcal{X}_B$, a probability kernel over $\mathcal{X}_A$ given $\mathcal{X}_B$ is a non-negative function
  \begin{equation}\label{def:kernel}
    q:\mathcal{X}_A\times\mathcal{X}_B\rightarrow\mathbb{R}_{+},
  \end{equation}
such that $\forall x_B\in\mathcal{X}_B,\sum_{x_A\in\mathcal{X}_A}q(x_A|x_B)=1$. Moreover, we require $q(x_A|x_B) \equiv \sum_{x_C \in \mathcal{X}_C} q(x_A, x_C|x_B)$ and $q(x_A|x_B, x_C) \equiv q(x_A, x_C|x_B) / q(x_C|x_B)$. In other words, a probability kernel behaves almost like a conditional probability, except that it bypasses the marginal distribution on the indexing set $\mathcal{X}_B$.

Within a given physical theory under the GPT framework, an observer accesses a physical state by performing measurements or effects on it and reading out the measurement outcomes. The measurement choice is denoted as a random variable, and the measurement outcome is described as a probability kernel over a fixed output alphabet given the measurement choice.
The probability kernel must satisfy certain rules, which stem from the laws on the valid physical states and operations defining the physical theory. These laws must observe a minimal set of assumptions, as described in Box~\ref{box:GPTAssumption}.

\begin{mybox}[label={box:GPTAssumption}]{{Assumptions in GPTs}}
{\centering \hfill(Following~\citet{barrett2007information})\hfill}
\par\noindent\emph{Single-party system}
\begin{enumerate}[leftmargin=*,nosep]
  \item The state of a single system can be determined by a probability kernel $p$ with respect to a complete set of fiducial measurements $\mathcal{F}$.
  \item The state set $\mathcal{S}$ is the convex hull of the allowed normalized states and a zero state $\vec{0}$.
  \item The set of allowed state transformations $\mathcal{T}$ is formed by \emph{all} linear maps which transform states into states. 
\end{enumerate}
\begin{samepage}
\emph{Multipartite systems}
\begin{enumerate}[leftmargin=*,nosep]\setcounter{enumi}{3}
  \item\emph{Parallel composition:} If $p_A\in\mathcal{S}_A$ and $p_B\in\mathcal{S}_B$, then $p_A \otimes p_B \in \mathcal{S}_{AB}$.
  \item\emph{No-signalling from the future:} 
  Given two parallel systems, under the action of any measurement on the first system, summing up the conditional states of the second system weighted by their probabilities leads to a state, which does not depend on the measurement applied to the first system.
\end{enumerate}
\end{samepage}
\end{mybox}

A particular choice of $\mathcal{S}$ and $\mathcal{T}$ together with possibly additional assumptions specifies a GPT. Note that measurements in the set of $\mathcal{F}$ are special physical operations belonging to $\mathcal{T}$.

\begin{remark}
    In the third assumption, we assume every possible linear map to be a valid physical transformation. This assumption is also called the \emph{no restriction hypothesis for transformations}. Consequently, this assumption implies the following properties:
\begin{enumerate}[leftmargin=*,nosep]
  \item The set of allowed transformations is closed under sequential composition.
  \item Transformations consisting of selecting a linear map on one system conditional on the outcome of a measurement of a parallel system are included.
  \item Every logically-possible effect is included in the effect space, a.k.a. the \emph{no restriction hypothesis for effects}.
\end{enumerate}
\end{remark}

Among the assumptions in Box~\ref{box:GPTAssumption}, Assumptions 4 and 5, which correspond to how multiple systems compose into a joint one, shall be essential to our results. In particular, Assumption 5 captures the essence of relativity, where the light speed poses a fundamental limit to information transmission. Intuitively, this assumption indicates that the state on the second system is invariant under the action of any measurement on the first system after ``forgetting" the outcome of the measurement on the first system. If this assumption is violated, namely, local operations on different parts of a joint system are not mutually commuting, then the disturbance by an operation at one place could be instantaneously detected by a distant part of the system. This indicates information transmission faster than the light speed and violates the rules of relativity. In Appendix~\ref{app:DAG}, we shall further explain the assumptions in Box~\ref{box:GPTAssumption}.

The GPT framework embeds the Newtonian and quantum physics. Going beyond, one of the most well-studied GPTs is the \emph{BoxWorld} theory~\cite{popescu1994quantum,barrett2007information,gross2010all}. In brief, BoxWorld is an explicit GPT capable of generating the most nonlocal statistics in Bell tests subjected to relativity. In the simplest case of the bipartite Bell scenario~\cite{clauser1969proposed},
consider two remote non-communicating observers, Alice and Bob. The two observers share a physical state, and they each apply random measurements on their share of the state. We denote their measurement choices as random variables $X$ and $Y$ and their measurement outcomes as variables $A$ and $B$, respectively. We denote the probability that the relation $X\cdot Y=A\oplus B$ is met as $S$. 
As there is no communication between the observers, each observer's measurement result is independent of the other's measurement choice, $X\independent B$ and $Y\independent A$, or written as
\begin{equation}\label{eq:nosignal}
\begin{split}
    p(a|x)&=\sum_b p(a,b|x,y)
    =\sum_b p(a,b|x,y'), \\
    p(b|y)&=\sum_a p(a,b|x,y)
    =\sum_a p(a,b|x',y). \\
\end{split}
\end{equation}
Consider the simple case where the inputs and outputs all take binary values in $\{0,1\}$. When the physical processes are subjected to classical physics, there is an upper bound of $S\leq3/4$~\cite{clauser1969proposed}. Quantum theory allows a higher value of it, exhibiting nonlocal correlations among the variables. However, $S$ is upper-bounded by $S\leq(2+\sqrt{2})/4$, still strictly away from unity~\cite{cirel1980quantum}. The BoxWorld theory closes this gap. Assume all the probability distributions $p(a,b|x,y)$ satisfying the constraints in Eq.~\eqref{eq:nosignal} represent valid physical processes on a bipartite state. Then, the following probability distribution represents a no-signalling physical process yielding a perfect correlation of $S=1$:
\begin{equation}\label{eq:PRBox}
    p(a,b|x,y)=\frac{1}{2}\delta_{a\oplus b=xy}.
\end{equation}
In the literature, this probability distribution is usually cited as the Popescu-Rohrlich (PR) box~\cite{popescu1994quantum}, which we denote as $p_{\mathrm{PR}}(a,b|x,y)$.

As mentioned above, DAGs provide a useful tool to visualize the causal structure among various variables in an event~\cite{pearl1995causal}. 
The parental vertex in an edge represents the cause, and the child vertex represents the effect. In Figure~\ref{fig:Bell}(a), we depict the DAG for the CHSH Bell test.
In classical causality, all the vertices, including observable and latent variables, are characterized as random variables, and causal influences are fully characterized by conditional probabilities. In this case, the causal model is also called a Bayesian network.
In a quantum or more general GPT causal theories, the latent variables correspond to the physical states in the underlying physical theory~\cite{henson2014theory}. Correspondingly, an observable variable with an incoming edge from the latent variable should be interpreted as the outcome of a measurement. Besides, a parental observable variable can serve as a classical control for the measurements.
In Appendix~\ref{app:DAG}, we shall rigorously state the physical interpretation of causal DAGs in GPTs.
Given a causal structure $G$, when the underlying physical theory is the Newtonian physics, we denote the set of valid probability distributions among observed variables as $\C(G)$.\elie{\footnote{Although all the example \emph{figures} in this manuscript have latent variables corresponding exclusively to parentless vertices, it should be noted that our \emph{results} apply to \emph{all} causal structures, including those with non-exogenous vertices corresponding to latent variables.}}
Analogously, given any physical theory within the GPT framework, we can obtain a set of valid probability distributions, whose union is denoted as $\GPT(G)$.
The BoxWorld theory is an oft-referenced \emph{particular} GPT; we denote the set of probability distributions consistent with the BoxWorld theory as $\BW(G)$. 
As shown in the CHSH Bell test, we have the strict inclusion $\C(G)\subsetneq\BW(G)$ in general.

To characterize plausible correlations in a causal structure, we are first interested in the independence conditions among observed random variables. 
In a Bayesian network, all the conditional independences can be enumerated from the Bayesian rules of probability calculation.
For instance, in Figure~\ref{fig:Verma}(a), as the causal influence of $X$ on $B$ is mediated by $A$, we have the conditional independence of $X\independent B|A$, or the Markov condition:
\begin{equation}
    p(x,a,b)=p(x)p(a|x)p(b|a).
\end{equation}
Alternatively, DAGs provide a graphical criterion to specify these conditional independencies, equivalently stated as the $d$-separation~\cite{pearl1995causal}. In brief, denote a pseudo path between two vertices as a sequence of edges linking them without two consecutive ones directing towards the same intermediate vertex. Here, the edges and vertices may involve the latent variables. Two vertices are $d$-separated from each other if they are not linked by a pseudo path, or $d$-separated by a set of vertices if the latter block all their linking pseudo paths.
While originally developed for the study of Bayesian networks, the $d$-separation holds for generalized latent variables to tell the statistical conditional independence conditions between observable variables~\cite{henson2014theory}.
For a causal structure characterized by a DAG $G$, we denote the set of probability distributions of observable variables whose conditional independence statements are consistent with the $d$-separation statements in $G$ as $\IND(G)$. 
In ~\citet{henson2014theory}, it has been proved that $\GPT(G)\subseteq\IND(G)$. In other words, the conditional independence conditions among observed variables are theory-independent.

\subsection{Nested Markov model in generalized probabilistic theories}\label{Sec:Nested}
As an equality constraint, each conditional independence condition in a causal structure, or equivalently, a $d$-separation condition, reduces a dimension of the statistical manifold for the set of valid distributions. 
However, in the study of Bayesian networks, it is found that there are other equality constraints in general. Consider the causal structure of the mediation test in Figure~\ref{fig:Verma}(a) and assume $\Lambda$ is a classical random variable for now. Using the Bayesian formula, one can prove the following expression,
\begin{equation}\label{eq:Verma}
    \partial_x \sum_{a}p(a|x)p(c|x,a,b) = 0,
\end{equation}
which means that the defined quantity ${q(c|b) \coloneqq \sum_{a}p(a|x)p(c|x,a,b)}$ is independent of $X$. Moreover, $q(c|b)$ is a probability kernel. This independence relation is called the Verma constraint~\cite{robins1986new,verma1990equivalence}. Different from the usual conditional independence conditions between random variables, the Verma constraint cannot be identified through the $d$-separation statements in $G$.

\begin{figure}[hbt!]
 \centering \includegraphics[width=\columnwidth]{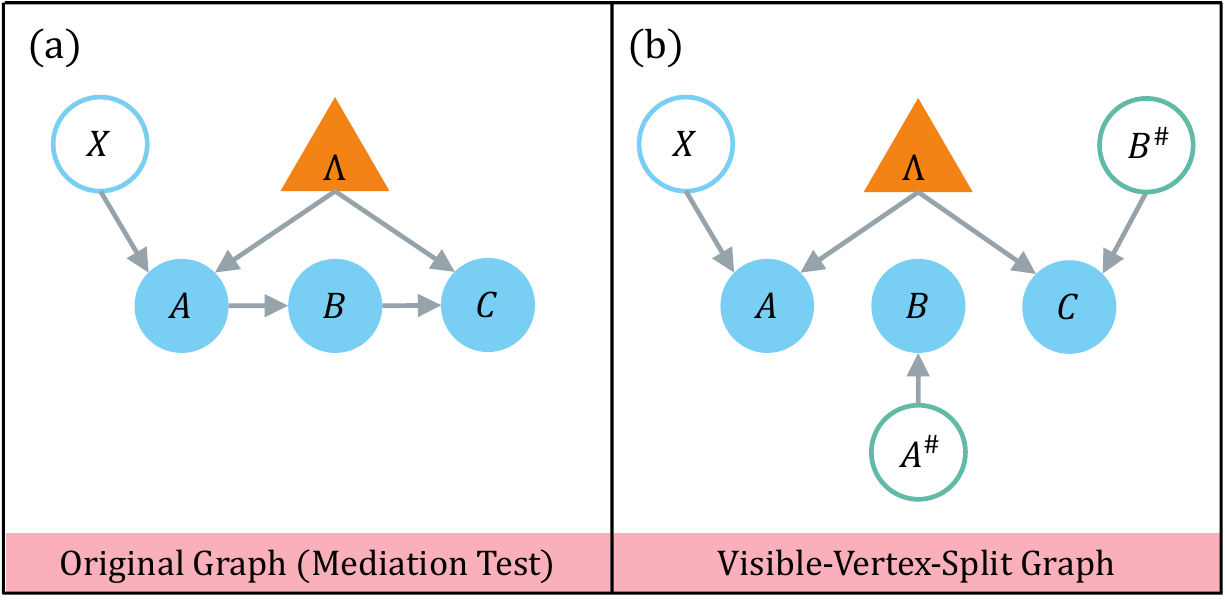}
\caption{Analysis for the mediation test scenario. (a) DAG of the mediation test. The causal structure includes observable variables $X,A,B$, and $C$, and a latent variable $\Lambda$. This is the simplest structure exhibiting the Verma constraint. 
(b) The visible-vertex-split graph. The green hollow circles are added vertices by vertex splitting, and a data post-selection of $A^\#{=}A$ and $B^\#{=}B$ projects the structure back to the mediation test.}
\label{fig:Verma}
\end{figure}

The original derivation of Verma constraints is restricted to Bayesian networks~\cite{robins1986new,verma1990equivalence}. In those proofs, besides the observed variables, the latent variables are also characterized as classical random variables. Verma constraints are obtained by decomposing the overall distribution of all the variables with respect to the usual Markov condition, applying the Bayes rule, and finally marginalizing over the distribution of latent variables.
In classical causal theory, given a causal structure $G$, the set of probability distributions that satisfy all the usual conditional independence conditions and Verma constraints is considered as the nested Markov model associated with $G$~\cite{evans2017margins,richardson2023nested}, which we denote as $\NEST(G)$.
The nested Markov model can be formally defined as the set of distributions, in which the $d$-separation statements about their identifiable kernels are consistent with conditional independence statements in the kernel space~\cite{richardson2023nested}. We shall rigorously define the model in Appendix~\ref{app:NestedDef}.
The nested Markov model is proven to provide an algebraically complete description of the Bayesian network with latent variables, listing exactly all the equality constraints among observed variables~\cite{evans2017margins}. Namely, $\C(G)$ and $\NEST(G)$ define manifolds in the probability simplex that are equivalent up to the difference of inequality constraints. Since a classical causal structure usually poses additional constraints, we have $\C(G)\subseteq\NEST(G)$, and there are causal structures making the inclusion strict.

Is the assumption of classicality of the latents critical to the derivation of Verma constraints? Or, perhaps Verma constraints and even the nested Markov model remain valid for GPT causal theories? As suggested by the proven validity of $d$-separation in GPTs~\cite{henson2014theory}, the possibility that Verma constraints might be properties inherent to a causal structure \emph{independent of the underlying physical theory} is a tantalizing prospect. Here, we affirm that this is indeed the case.

\begin{theorem}[Nested Markov property in GPTs]
  Consider a causal structure defined by DAG $G$, where the latent variables represent physical states in some GPT. Any possible probability distribution over the observable vertices satisfies all the nested Markov constraints for $G$. Formally, $\GPT(G)\subseteq\NEST(G)$.
\label{thm:nested}
\end{theorem}

The proof of this theorem motivates a causal model standing between the nested Markov model and GPTs, which we shall defer to the middle of this manuscript.
Before we move on, we point out that this titular result comes along with a lovely corollary. In~\citet{evans2017margins}, it is proved that $\NEST(G)$ provides a complete characterization of the \emph{equalities} in $\C(G)$. That is, Evans shows that the Zariski closures of $\NEST(G)$ and $\C(G)$ coincide. A consequence of Theorem~\ref{thm:nested}, then, is:
\begin{corollary}
For any DAG $G$, $\NEST(G)$ provides a complete characterization of the \emph{equalities} in $\GPT(G)$, as the Zariski closures of 
$\C(G)$ and $\GPT(G)$ and $\NEST(G)$ all coincide.
\end{corollary}

To have a flavour of why the nested Markov model remains valid for GPTs, we describe a purely graphical procedure and specify the Verma constraint of Eq.~\eqref{eq:Verma} directly from the causal DAG in Figure~\ref{fig:Verma}(a), though this graphical procedure was originally developed in the context of Bayesian networks~\cite{tian2002general,tian2002testable}.
In general, given a causal DAG $G$ with the set of vertices representing observed variables $V$ and the set of vertices representing latent variables $W$, graphically, we may decompose $V$ into districts.
That is, two vertices $v_i,v_j\in V$ are said to be bi-directed-linked if there exists a divergent path of $v_i\leftarrow w_k\cdots w_l\rightarrow v_j$, such that $w_k,w_l\in W$ and there is a path between $w_k$ and $w_l$ going strictly within $W$ \cite{tian2002general}.
Then, following \citet{evans2017margins}, we define a district of vertices representing observed variables as the maximal inclusion set of vertices that are bi-directed-linked.
For a district of vertices $\{v_i\}_i\subseteq V$ and their parental vertices that represent observed variables, $\mathrm{pa}_G(\{v_i\}_i)\backslash W$, we can specify a probability kernel over the state space of the variables carried by $\{v_i\}_i$ given that of $\mathrm{pa}_G(\{v_i\}_i)\backslash W$. Here, we denote the parental vertices of a vertex $v$ in $G$ as $\mathrm{pa}_G(v)$ and similarly for a set of vertices.
For the causal structure in Figure~\ref{fig:Verma}(a), we can specify three districts among the vertices representing observed variables, given by $\{X\},\{B\}$, and $\{A,C\}$. Suppose there is a probability distribution $p(x,a,b,c)$ among the observed variables. In the study of Bayesian networks, it was shown that the distribution can be factorized as:
\begin{equation}\label{eq:factorization}
    p(x,a,b,c)=p(x)p(b|a)q(a,c|x,b),
\end{equation}
where $p(x),p(b|a)$, and $q(a,c|x,b)$ represent probability kernels corresponding to each district. 
We prove that the same factorization still holds under GPTs.
In these kernels, $p(x)$ and $p(b|a)$ are marginal distributions obtained from $p(x,a,b,c)$. Furthermore, $p(b|a)$ can be identified by the causal DAG with vertices $B$ and $A^\#$ in Figure~\ref{fig:Verma}(b), and $q(a,c|x,b)$ is a probability kernel over $\mathcal{X}_A$ and $\mathcal{X}_C$ given $\mathcal{X}_X$ and $\mathcal{X}_B$, which can be identified by the causal DAG with vertices $X,A,C,B^\#$, and $\Lambda$ in Figure~\ref{fig:Verma}(b). This vertex-split causal DAG has vertex $C$ $d$-separated from vertex $X$ (even) by vertex $B^\#$. Since for any causal DAG $G$ we know that $\GPT(G) \subseteq \IND(G)$, it therefore follows that this kernel satisfies 
\begin{equation}\label{eq:margin}
    q(c | b) = \sum_{a} q(a,c|x,b).
\end{equation}
Applying Bayes rule to the observed probability $p(x,a,b,c)$, one can easily check that 
\begin{equation}\label{eq:bayes}
    q(a,c|x,b)=\frac{p(x,a,b,c)}{p(x)p(b|a)}=p(a|x)p(c|x,a,b),
\end{equation}
which is the expression to be summed over in Eq.~\eqref{eq:Verma}.
In light of~\eqref{eq:margin} and~\eqref{eq:bayes}, we arrive at the Verma constraint in Eq.~\eqref{eq:Verma}.

\subsection{A first causal model between nested Markov and GPTs}\label{sec:comparison}

As a technical preparation to prove Theorem~\ref{thm:nested}, we must first distinguish between generic DAGs versus those more typically found in the study of Network Nonlocality \cite{tavakoli2022reviewnetworks}. Such DAGs, which we call network and semi-network type DAGs, are those admitting interpretations such that every observable variable can be understood uniquely as either a setting or a measurement outcome. Notions such as the set of \emph{no-signalling} correlations are only well-defined on these elementary-form causal structures.

\begin{definition}[\textbf{Network} and \textbf{Semi-Network} Type DAGs]
Consider a causal structure denoted by a DAG $G(V,W,\mathcal{E})$, where $V$ is the set of vertices representing observed variables, $W$ is the set of vertices representing latent variables, and $\mathcal{E}$ is the edge set. $G$ is said to be \textbf{semi-network} type if there does \emph{not} exist any vertex in $V$ with both (at least) one parent as well as (at least) one child. A semi-network type DAG is said to furthermore be \textbf{network} type if no vertex in $V$ has more than one child.
\end{definition}

Our first observation is that there is no \emph{opportunity} to further distinguish nested Markov models on (semi-) network-type DAGs from ordinary Markov models. As formally proven in Appendix~\ref{app:prop1}, a prerequisite for any nontrivial nested Markov constraints to be implied is that the graph must feature an observable-variable vertex with both at least one observable descendant and at least one observable ancestor. Semi-network-type DAGs, however, are exactly those \textit{without} such phenomenology. As such,
\begin{proposition}
    If $G$ is a network or semi-network-type DAG, then $\NEST(G)=\IND(G)$. 
\label{prop:already_seminetwork}
\end{proposition}
If $G$ is network-type and furthermore \emph{has only one latent variable}, then $\IND(G)=\BW(G)$, which is the set of correlations obtainable by the BoxWorld GPT~\cite{gross2010all}, also known as GNST~\cite{barrett2007information}.

Semi-network-type DAGs arise naturally in the context of \textbf{Single World Intervention Templates} (SWITs) introduced by~\citet{richardson2013single}. SWITs are graphs formed by performing a \emph{splitting} operation on vertices corresponding to observable variables whenever the vertex in question has both a parent and a child. We begin by recapitulating the \emph{visible vertex splitting} construction behind SWITs. We then explore the causal model induced by positing the existence of a consistent assignment of counterfactual probabilities (a.k.a. potential outcomes~\cite{rubin2005causal, shpitser2022multivariate}), which are \emph{ordinary} Markov with respect to the visible-vertex-split graph.

\begin{samepage}
\begin{definition}[The \textbf{Visible Vertex Splitting} Construction]\label{def:SWIG}
    Consider a causal structure denoted by a DAG $G(V,W,\mathcal{E})$, where $V$ is the set of vertices representing observed variables, $W$ is the set of vertices representing latent variables, and $\mathcal{E}$ is the edge set. To construct the \textbf{visible-vertex-split graph} of $G$, we proceed as follows:
    \par\nopagebreak
    \texttt{\textbf{Let}} $V_{\mathrm{split}}$ be the subset of vertices within $V$, which have both incoming and outgoing edges. That is, $v\in V_{\mathrm{split}}\subseteq V$ iff $\exists {(e: x\to v)}\in\mathcal{E}$ and $\exists {(e': v\to y)}\in\mathcal{E}$ for some $\{x,y\}\in V\cup W$.
    \par\nopagebreak
    \texttt{\textbf{For}} $v_i\in V_{\mathrm{split}}$, replace all the edges in $\mathcal{E}$ which \emph{originate} from $v_i$ by an edge which originates from $v^\#_{i}$ instead. 
    \par\nopagebreak
    Here, $v^\#_{i}$ is a newly-added root vertex equipped with a corresponding classical variable $X_{v^\#_i}$ sharing the same sample space as $X_{v_i}$.
    \par\nopagebreak
    The visible-vertex-split graph of $G$ is denoted as $\VVSplit(G)$. Denoting the set of its vertices of the form $v^\#_{\_}$ by $V^\#$, and denoting its (modified) edge set as $\mathcal{E}^{'}$, we say that the visible vertex splitting operation takes the causal structure $G(V,W,\mathcal{E})$ to a new causal structure, $\VVSplit(G)(V{\cup}V^\#,W,\mathcal{E}')$.
\end{definition}
\end{samepage}

Briefly speaking, the visible-vertex-split graph is constructed according to the following prescription: For any vertex that needs splitting, split it.

\begin{remark}\label{remark:VVSplit}
We observe the following properties of visible-vertex-split graphs:
\begin{enumerate}[leftmargin=*,nosep]
    \item The total number of edges in $\mathcal{E}$ and in $\mathcal{E}'$ are the same. The visible vertex splitting only \emph{modifies} edges.
    \item  Given a visible-vertex-split graph $G'(V{\cup}V^\#,W,\mathcal{E}')$ and $v^\#_i\in V^\#$, $v^\#_i$ has \emph{no parents} in $G'$.
    \item Visible vertex splitting is stable under iteration. That is, $\VVSplit(G) =  \VVSplit\bigl(\VVSplit(G)\bigr)$ for any starting DAG $G$.
    \item Every visible-vertex-split graph is a \emph{semi-network} type DAG.     
\end{enumerate}
\end{remark}

We give an example in Figure~\ref{fig:Verma}(b), where $A\rightarrow B$ is replaced with $A^\#\rightarrow B$, and $B\rightarrow C$ is replaced with $B^\#\rightarrow C$, with $A^\#$ and $B^\#$ being the newly added vertices. With the additional vertices introduced in the definition of $\VVSplit(G)$, every observed vertex is now either root or terminal. This makes $\VVSplit(G)$ a semi-network-type DAG.

As the reader may anticipate, the correlation set of $\VVSplit(G)$ should be closely related with that of the original graph $G$. To formalize the statements, we first define a distribution projection operation equivalent to the concept of \emph{consistency} in~\citet{richardson2013single}.

\begin{definition}[The Visible-Vertex-Split Projection]
    Given a causal DAG $G(V,W,\mathcal{E})$ and its associated visible-vertex-split graph $\VVSplit(G)(V{\cup}V^\#,W,\mathcal{E}')$ obtained from Definition~\ref{def:SWIG}, let $q(X_V|X_{V^\#})$ be any probability kernel among the observed variables $X_V=\{X_{v_i}\}_i$, and $X_{V^\#}=\left\{X_{v^\#_{j}}\right\}_j$. Then, the \textbf{visible vertex splitting-projected} probability ${p(X_V) \coloneqq \Proj\bigl(q(X_V|X_{V^\#})\bigr)}$ is defined by conditioning on events where the variables in $X_{V^\#}$ take on the same values as those dictated by their corresponding variables within $X_V$, namely,
    \begin{align} 
    &p\left(\bigwedge_{\;\mathclap{v_i \in V}} X_{v_i}{=}x_{v_i}\right)
    = q\left(\bigwedge_{\;\;\mathclap{v_i \in V}} X_{v_i}{=}x_{v_i}\bigg\vert\bigwedge_{\;\mathclap{\substack{v_i\in V:\\ {v^\#_i}\in V^\#}}} X_{{v^\#_i}}{=}x_{v_i}\right).
    \end{align}
\label{def:projection}
\end{definition}

We remark that ${{v^\#_{i}}\in {V^\#}}$ only if ${{v_{i}}\in {V}}$ but not necessarily \emph{vice versa}. It can be easily checked that $p=\Proj(q)$ is also a probability kernel. 
The utility of the projection is that it relates the correlations over $X_V$ in any causal DAG to the correlations over ${X_V{\cup}X_{V^\#}}$ in its visible-vertex-split graph for any GPT, as shown by the following proposition.

\begin{proposition}\label{prop:projection}
    Given a GPT, $\mathcal{G}$, and a causal DAG, $G(V,W,\mathcal{E})$, denote $\mathcal{G}(G)$ as the set of valid distributions among observed variables of $G$ within $\mathcal{G}$. A distribution $p(X_V)$ is within $\mathcal{G}(G)$ if and only if there exists some probability kernel $q(X_V|X_{V^\#})$, which is within the set of distributions $\mathcal{G}\bigl(\VVSplit(G)\bigr)$, such that $p$ is the projection of $q$. Formally, $p(X_V) \in \mathcal{G}(G)$ iff
    \begin{align*}
    \exists_{q(X_V|X_{V^\#})\in\mathcal{G}\left(\VVSplit(G)\right)}\text{ s.t. }p(X_V)=\Proj\bigl(q(X_V|X_{V^\#})\bigr).
    \end{align*}
\end{proposition}

The proof is quite intuitive. Since we only add vertices corresponding to observed variables, the effect on any original vertex in $G$ is not changed. Specifically, the outgoing systems of each latent variable remain the same; for each observed vertex $v$ in $G$, its incoming systems remain the same; and for every measurement choice, the physical measurement performed on the system is the same. We leave the formal proof in Appendix~\ref{app:prop2}.

Now, let us inspect the operations in Definitions~\ref{def:SWIG} and~\ref{def:projection}. Physically, $\IND(\VVSplit(G))$ is defined via conditional independence constraints, thereby representing the minimal requirement to define the physical rules of calculating a probability kernel in consistency with relativity principles. Recall that pursuant to Proposition~\ref{prop:already_seminetwork}, $\IND\bigl(\VVSplit(G)\bigr)=\NEST\bigl(\VVSplit(G)\bigr)$. This universality encourages us to define a causal model consisting of the projection of the set of kernels that are ordinary Markov with respect to the visible-vertex-split graph. 

\begin{definition}[The Visible-Vertex-Split Model]\label{def:ssg}
   Given a causal DAG $G$, define its \textbf{visible-vertex-split model} as $\VVS(G)\coloneqq \Proj\bigl(\IND\bigl(\VVSplit(G)\bigr)\bigr)$.
   Formally, $p(X_V) \in \VVS(G)$ iff
    \begin{align*}
    \exists_{q(X_V|X_{V^\#})\in\mathcal{I}\left(\VVSplit(G)\right)}\text{ s.t. }p(X_V)=\Proj\bigl(q(X_V|X_{V^\#})\bigr).
    \end{align*}
\end{definition}

Since \citet{henson2014theory} already proved that GPT-realizable correlations are contained in the set of ordinary Markov models for any DAG, the following lemma follows immediately:

\begin{lemma}\label{lemma:easy_part}
  For any DAG $G$, ${\GPT(G)\subseteq\VVS(G)}$.
\end{lemma}

Our titular claim that ${\GPT(G)\subseteq\NEST(G)}$ follows from the following fact.

\begin{lemma}\label{lemma:hard_part}
  For any DAG $G$, ${\VVS(G)\subseteq\NEST(G)}$.
\end{lemma}
Lemma~\ref{lemma:hard_part} was implicitly indicated by~\citet{richardson2013single} (Proposition~17), though we provide our own proof in Appendix~\ref{app:lemma1}. Lemmas~\ref{lemma:easy_part} and~\ref{lemma:hard_part} together cumulatively imply Theorem~\ref{thm:nested}.

We can also recognize that for some DAGs, $\VVS(G)=\NEST(G)$. Notably, for all \emph{semi-network} type DAGs,  
\begin{align*}
\VVS(G)=\NEST(G)=\IND(G),
\end{align*}
since semi-network type DAGs have the property of being invariant under visible vertex splitting.

While the nested Markov model fully \emph{covers} the predictions of GPTs, we now show that the converse is not true. Namely, there exist causal structures such that $\GPT(G)\subsetneq\NEST(G)$. This should not be surprising, as the nested Markov model is defined solely by equality constraints in the probability space. On the other hand, the requirement of ``physical'' states and operations in a GPT leads to non-trivial inequality constraints. In other words, GPTs do not provide the exact physical interpretation of the nested Markov model. 
A paradigmatic example of such a graph is the renowned \emph{Instrumental} scenario~\cite{wright1928tariff} shown in Figure~\ref{fig:Bell}(b). As previously shown by \citet{henson2014theory,van2019quantum}, all distributions within the visible-vertex-split model of the Instrumental scenario, where the visible vertex split graph is the Bell test scenario in Figure~\ref{fig:Bell}(a), must satisfy Pearl's instrumentality inequality~\cite{pearl1995testability}:
\begin{equation}\label{eq:instrument}
    \max_a\sum_b\max_x p(a,b|x)\leq 1.
\end{equation}
On the other hand, the Instrumental scenario does not exhibit any equality constraints, either the ordinary conditional independence or Verma-type constraints. Consequently, the models of $\NEST(G_{\mathrm{ins}})$ and $\IND(G_{\mathrm{ins}})$ saturate the entire probability space for $G_{\mathrm{ins}}$ being the causal DAG of the Instrumental scenario. Hence, we have the strict inclusion of $\VVS(G_{\mathrm{ins}})$ within $\NEST(G_{\mathrm{ins}})$, and $\GPT(G_{\mathrm{ins}})\subsetneq\NEST(G_{\mathrm{ins}})$ follows.
Incidentally, \citet{van2019quantum} additionally demonstrates that $\C(G) \subsetneq \GPT(G)$. These results together demonstrate that the Instrumental scenario leads to a fine-grained hierarchy of various causal models.

One way to understand that a wide class of graphs where $\VVS(G)$ is strictly contained within $\NEST(G)$---including the causal DAG of the Instrumental scenario---is to appreciate that $\VVS(G)$ is always constrained by the principle of \(e\)-separation, and this principle is the origin of some non-trivial inequality constraints similar to Pearl's instrumentality inequality~\cite{evans2012graphical}. 
As the nested Markov model overlooks inequality constraints, consequently, $\VVS(G)\subsetneq\NEST(G)$ for any graph which exhibits non-trivial \(e\)-separation relations.
We refer the readers to~\citet{finkelstein2021entropic} for a discussion on deriving inequality constraints for generic DAGs via the principle of $e$-separation. In Appendix~\ref{app:esep}, we review the $e$-separation criterion and prove why this principle follows from the principle of visible vertex splitting.

\subsection{A second causal model between nested Markov and GPTs}\label{sec:maxint}

Having established that ${\GPT(G)\subseteq\VVS(G)}$, we next inquire if that inclusion is strict, or even perhaps the two sets coincide.
Indeed, there do exist some causal DAGs where GPT models \emph{saturate} the set $\VVS(G)$, with the Instrumental scenario being such an example. Actually, any causal DAG $G$ with a single latent variable, and where furthermore no observable-variable vertex has \emph{multiple children}, is such that $\VVS(G)=\BW(G)$, since $\VVSplit(G)$ is a network-type DAG that corresponds to a (multipartite) Bell test scenario, and $\IND\bigl(\VVSplit(G)\bigr))=\BW\bigl(\VVSplit(G)\bigr))$. Notably, as BoxWorld is one \emph{particular} GPT, we know that $\BW(G)\subseteq \GPT(G)$. Putting these arguments together, we thus have $\GPT(G)=\VVS(G)$ for such causal DAGs.

That said, there are plenty of DAGs where ${\GPT(G)\subsetneq\VVS(G)}$. We begin to illustrate this point by constructing a second causal model which subsumes all GPT-realizable distributions but which, for some DAGs, is strictly contained within $\VVS(G)$.

To introduce this next causal model, we first define operations of a causal DAG extension known as \emph{Maximal Interruption} and a projection postulate (a.k.a. the \emph{Interruption Lemma}) allowing us to define models on the original graph by their relationship to ordinary Markov models in the associated interruption graph. Note that maximal interruption graphs were introduced by \citet{wolfe2021quantum} as means to upgrade a computational method for characterizing the set of \emph{quantum} correlations with a network into a method suitable for analyzing more general DAGs. Our approach is similar, except that instead of focusing on \emph{quantum} correlations, we essentially upgrade the network-specific concept of no-signalling into a causal model applicable to general DAGs.

\begin{samepage}
\begin{definition}[The \textbf{Maximal Interruption} Construction]\label{def:hypergraph}
    Consider a causal structure denoted by a DAG $G(V,W,\mathcal{E})$, where $V$ is the set of vertices representing observed variables, $W$ is the set of vertices representing latent variables, and $\mathcal{E}$ is the edge set. To construct the \textbf{maximal interruption graph} of $G$, we proceed as follows:
    \par\nopagebreak
    \texttt{\textbf{Let}} $V_{\mathrm{split}}$ be the subset of vertices within $V$ which have both incoming and outgoing edges. That is, $v\in V_{\mathrm{split}}$ iff $\exists {(e: x\to v)}\in\mathcal{E}$ and $\exists {(e': v\to y)}\in\mathcal{E}$ for some $\{x,y\}\in V\cup W$.
    \par\nopagebreak
    \texttt{\textbf{Let}} $V_{\mathrm{rip}}$ be the subset of vertices within $V$ each of which has more that one outgoing edges. That is, $v\in V_{\mathrm{rip}}$ iff $\left|\{e_j:\textsf{origin}(e_j)=v\}\right|\geq 2$.
    \par\nopagebreak
    \texttt{\textbf{For}} $v\in V_{\mathrm{split}}\cup V_{\mathrm{rip}}$, replace each of its outgoing edges $e_j$ with an edge that originates from $u_{e_j}$ instead, where $u_{e_j}$ is a newly-added root vertex equipped with a corresponding classical variable $X_{u_{e_j}}$ sharing the same sample space as $X_v$. 

     The maximal interruption graph of $G$ is denoted as ${\MaxInt(G)}$. Denoting the set of its vertices of the form $u_{e}$ by $U$, and denoting its (modified) edge set as $\mathcal{E}'$, we say that the maximal interruption operation takes the causal structure $G(V,W,\mathcal{E})$ to a new causal structure $\MaxInt(G)(V{\cup}U,W,\mathcal{E}')$.
\end{definition}
\end{samepage}

The maximal interruption graphs share similar properties as those of visible-vertex-split graphs in Remark~\ref{remark:VVSplit}, except a critical difference in dealing with the vertices representing observed variables with multiple child vertices in the original causal structure.

\begin{remark}
    It is worth contrasting the maximal interruption construction here with the visible-vertex-splitting construction underlying SWITs. Unlike SWITs, maximal interruption introduces new auxiliary vertices for \emph{each} child of a vertex in the original DAG.
    Consequently, every maximal interruption graph is a \emph{network} type DAG.
\end{remark}

We can now define distribution projection operation for maximal interruption graphs following the same idea as in Definition~\ref{def:projection}.

\begin{definition}[The Maximal Interruption Projection]\label{def:projectionMI}
    
    Given a causal DAG $G(V,W,\mathcal{E})$ and its associated maximal interruption graph $\MaxInt(G)(V{\cup}U,W,\mathcal{E}')$ obtained from Definition~\ref{def:hypergraph}, let $q(X_V|X_U)$ be some valid probability kernel among the observed variables $X_V=\{X_{v_1},\cdots,X_{v_n}\}$, and $X_U=\{X_{u_{e_1}},\cdots,X_{u_{e_k}}\}$, where $e_i$ iterates over the set of $k$ edges within $\mathcal{E}$ which are outgoing from vertices in $V$. It is convenient to partition the nodes in $X_U$ according to their corresponding origin-of-edge vertex in $\mathcal{E}$. That is, we write 
    \begin{align}
      X_U = \bigcup_{v_i\in V} \bigcup_{\substack{e_j\in \mathcal{E}:\\\textsf{origin}(e_j)=v_i}} X_{u_{e_j}}.
    \end{align}
    Then, the \textbf{maximal interruption-projected} probability ${p(X_V) \coloneqq \Proj\bigl(q(X_V|X_U)\bigr)}$ is defined by conditioning on events where the variables in $X_U$ take on the same values as those dictated by the variables in $X_V$, using origin-of-edge vertex as the mapping correspondence.
    That is,
    \begin{align}
      &p\left(\bigwedge_{\;\mathclap{v_i \in V}} X_{v_i}{=}x_{v_i}\right)
    \\\nonumber &
    = q\left(\bigwedge_{\;\mathclap{v_i \in V}} X_{v_i}{=}x_{v_i}\bigg|\bigwedge_{v_i \in V} \bigwedge_{\substack{e_j\in \mathcal{E}:\\\mathclap{\textsf{origin}(e_j)=v_i}}} X_{u_{e_j}}{=}x_{v_i}\right)
    \end{align}
\end{definition}
For simplicity, we use the same notation for the maximal interruption projection as for the visible vertex splitting projection. It can be easily checked that $p=\Proj(q)$ is also a probability kernel. 
The utility of the projection is that is relates the correlations over $X_V$ in any causal DAG to the correlations over ${X_V{\cup}X_U}$ in its maximal interruption graph for any GPT. Since it is more straightforward to describe a GPT in terms of its predictions with respect to network-type DAGs, this projection can be leveraged to translate intuitions about GPTs in network-type DAGs into results pertinent to non-network-type DAGs. We here reproduce the \emph{Fundemental Lemma of Interruption} from \citet{wolfe2021quantum} in the notation of this article:

\begin{proposition}\label{prop:projectionMI}
    Given a GPT, $\mathcal{G}$, and a causal DAG, $G(V,W,\mathcal{E})$, a distribution $p(X_V)$ is within the set of correlations $\mathcal{G}(G)$ if and only if there exists some probability kernel $q(X_V|X_U)$, which is within the set of correlations $\mathcal{G}\bigl(\MaxInt(G)\bigr)$, such that $p$ is the projection of $q$. Formally, $p(X_V) \in \mathcal{G}(G)$ iff
    \begin{align*}
    \exists_{q(X_V|X_U)\in\mathcal{G}\left(\MaxInt(G)\right)}\text{ s.t. }p(X_V)=\Proj\bigl(q(X_V|X_U)\bigr).
    \end{align*}
\end{proposition}

Note that Proposition~\ref{prop:projectionMI} is nearly identical to Proposition~\ref{prop:projection}, except that it has been upgraded to apply to maximal interruption graphs. The proof of Proposition~\ref{prop:projectionMI} is essentially the same as the proof of Proposition~\ref{prop:projection} by noticing that only the vertices corresponding to observable variables are being intervened upon. The actual physical states represented by the latent vertices and transformations (including effects) represented by them together with their child vertices are not changed. We leave the formal proof in Appendix~\ref{app:prop2}.

Now, let us inspect the operations in Definitions~\ref{def:hypergraph} and~\ref{def:projectionMI}. Physically, $\IND(\MaxInt(G))$ is defined via conditional independence constraints in a network-type DAG, thereby representing the minimal physical requirement of the no-signalling condition for a valid observable probability distribution.
Recall that pursuant to Proposition~\ref{prop:already_seminetwork}, $\IND\bigl(\MaxInt(G)\bigr)=\NEST\bigl(\MaxInt(G)\bigr)$. This universality encourages us to define a causal model consisting of the projection of the set of kernels that are ordinary Markov with respect to the maximal interruption graph. 

\begin{definition}[The Maximal Interruption Model]
   Given a causal DAG $G$, define its \textbf{maximal interruption} model as $\MI(G)\coloneqq \Proj\bigl(\IND\bigl(\MaxInt(G)\bigr)\bigr)$.
   Formally, $p(X_V) \in \MI(G)$ iff
    \begin{align*}
    \exists_{q(X_V|X_U)\in\mathcal{I}\left(\MaxInt(G)\right)}\text{ s.t. }p(X_V)=\Proj\bigl(q(X_V|X_U)\bigr).
    \end{align*}
\label{def:MIset}
\end{definition}

Using Propositions~\ref{prop:already_seminetwork} and~\ref{prop:projectionMI} it should be immediately clear that

\begin{lemma}\label{lemma:MIproject}
For any DAG $G$, ${\GPT(G)\subseteq\MI(G)\subseteq\VVS(G)}$.
\end{lemma}

Compared to the causal model defined through the visible-vertex-split graph, $\VVS(G)$, there are additional constraints in $\MI(G)$ that arise only through the maximal interruption construction and projection. That is, there are causal structures where $\MI(G)\subsetneq\VVS(G)$.\footnote{Of course, there are also causal structures where $\MI(G)=\VVS(G)$, such as any graph where visible vertices never exceed more than one child.} Such causal structures include the ones that demonstrate the phenomenon of the monogamy of nonlocality \cite{augusiak2014elemental}. Here, we provide a simple example inspired by a recent result by \citet{centeno2024significance}. Consider the causal DAG $G_{\mathrm{exo}}$ shown in Figure~\ref{fig:MaxInterrupt}(a), which represents the causal structure called the exogenized scenario. This causal DAG contains observable variables $X,Y,A,B,C$ and a latent variable $\Lambda$. It can be easily checked that $G_{\mathrm{exo}}$ is a semi-network-type DAG and hence $\VVSplit(G_{\mathrm{exo}})=G_{\mathrm{exo}}$. Nevertheless, since $Y$ has multiple child vertices, its maximal interruption graph is non-trivial, which is shown in Figure~\ref{fig:MaxInterrupt}(b) with additional vertices $U_{Y\too C}$ and $U_{Y\too B}$ added as per the edges directed from $Y$ to $B$ and $C$ in $G_{\mathrm{exo}}$, respectively. Now consider the following distribution:
\begin{equation}\label{eq:twolayerdist}
    p(a,b,c|x,y)=p_{\mathrm{PR}}(a,b|x,y)\delta_{b=c},
\end{equation}
with $p_{\mathrm{PR}}(a,b|x,y)$ representing the PR box given in Eq.~\eqref{eq:PRBox}.
Actually, one can consider any bipartite Bell-test distribution that violates the classical bound of a Bell inequality in the above distribution.
This distribution does not violate any conditional independence constraint in Figure~\ref{fig:MaxInterrupt}(a), namely, it belongs to the causal model $\VVS(G_{\mathrm{exo}})$. However, in the maximal interruption graph, we can specify additional non-trivial independence constraints after new vertices are added. For instance, for any valid distribution $p(a,b,c|x, u_{Y\too C}, u_{Y\too B})$ in $\MI(G_{\mathrm{exo}})$, we have the constraint
\begin{equation}
\begin{split}
    &p(a,c|x,u_{Y\too C}, u_{Y\too B})=p(a,c|x,u_{Y\too C}).
\end{split}
\end{equation}
After considering all the non-trivial independence constraints and applying the maximal interruption projection, it can be checked that the distribution in Eq.~\eqref{eq:twolayerdist} is not in $\MI(G_{\mathrm{exo}})$ and hence not realizable in the causal structure of $G_{\mathrm{exo}}$ in any GPT \cite{centeno2024significance}.

\begin{figure}[hbt!]
\centering
 \includegraphics[width=\columnwidth]{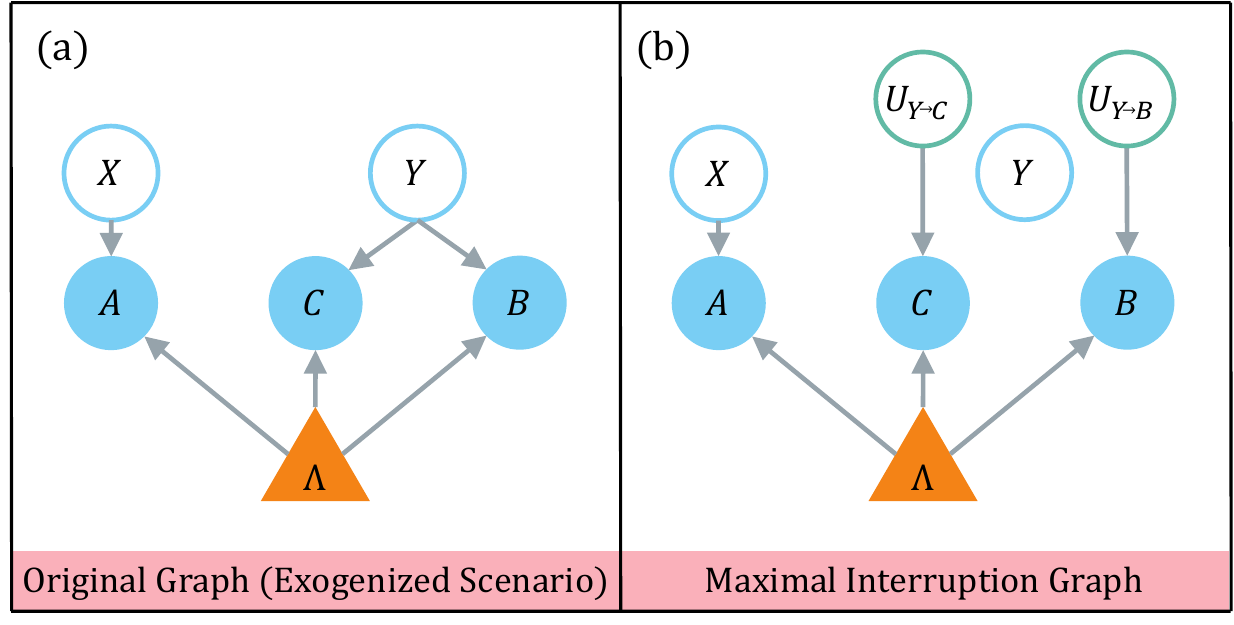}
\caption{Analysis for the exogenized scenario. (a) DAG of the exogenized scenario. In this causal structure, the observed variable $Y$ has multiple child vertices. (b) The maximal interruption graph. The green hollow circles are added vertices by maximal interruption, and a data post-selection of $U_{Y\too C}=U_{Y\too B}=Y$ projects the structure back to the exogenized scenario.}
\label{fig:MaxInterrupt}
\end{figure}

To conclude the discussion of maximal interruption, we specify the causal model $\MI(G)$ for certain causal DAGs. Firstly, for any DAG $G$ that is invariant under maximal interruption, its post-selection distribution set is the same as the ordinary Markov set. That is,
\begin{align*}
\text{For all \emph{network} type DAGs:}\quad   \MI(G)=\NEST(G)=\IND(G)
\end{align*}
since network-type DAGs have the property of being invariant under maximal interruption. Secondly, sometimes the maximal interruption principle exactly captures the set of GPT realizable distributions. The following proposition lists one such case.

\begin{proposition}
    If $G$ is a causal DAG with no more than one vertex corresponding to a latent variable, then $\GPT(G)=\MI(G)$. 
\label{prop:MINF_saturated}
\end{proposition}

The proof can be found in Appendix~\ref{app:prop_one_latent}.

\subsection{Nonfanout Inflation: constraining GPTs beyond maximal interruption}
In constructing hypergraphs that lead to the causal models analyzed thus far, we simply introduce auxiliary vertices that represent copies of observed variables. A natural following question is whether the maximal interruption model provides a tight characterization of the GPT-realizable distributions. For this purpose, consider the triangle scenario shown in Figure~\ref{fig:Triangle}(a), characterized by the causal DAG $G_{\mathrm{tri}}$. The causal structure is naturally invariant under maximal interruption, as it does not exhibit any directed edges originating from a vertex corresponding to an observed variable. At the same time, as the vertices in the triangle scenario are linked with each other nontrivially, we expect limitations on the strengths of the correlations that the triangle scenario can exhibit. Indeed, \citet{henson2014theory} proved that probability distribution
\begin{equation}\label{eq:synchr}
    p(a,b,c)=
    \begin{cases}
        \frac{1}{2}, & \mbox{if } a=b=c=0, \\
        \frac{1}{2}, & \mbox{if } a=b=c=1, \\
        0, & \mbox{otherwise}
    \end{cases}
\end{equation}
cannot be realized in the triangle scenario by any GPT. Evidently, then, $\GPT(G_{\mathrm{tri}})\subsetneq\MI(G_{\mathrm{tri}})$. Incidentally, the triangle scenario also exhibits ${\C(G_{\mathrm{tri}})\subsetneq\GPT(G_{\mathrm{tri}})}$ as per~\citet{fritz2012beyond}.

\begin{figure}[hbtp!]
\centering
 \includegraphics[width=\columnwidth]{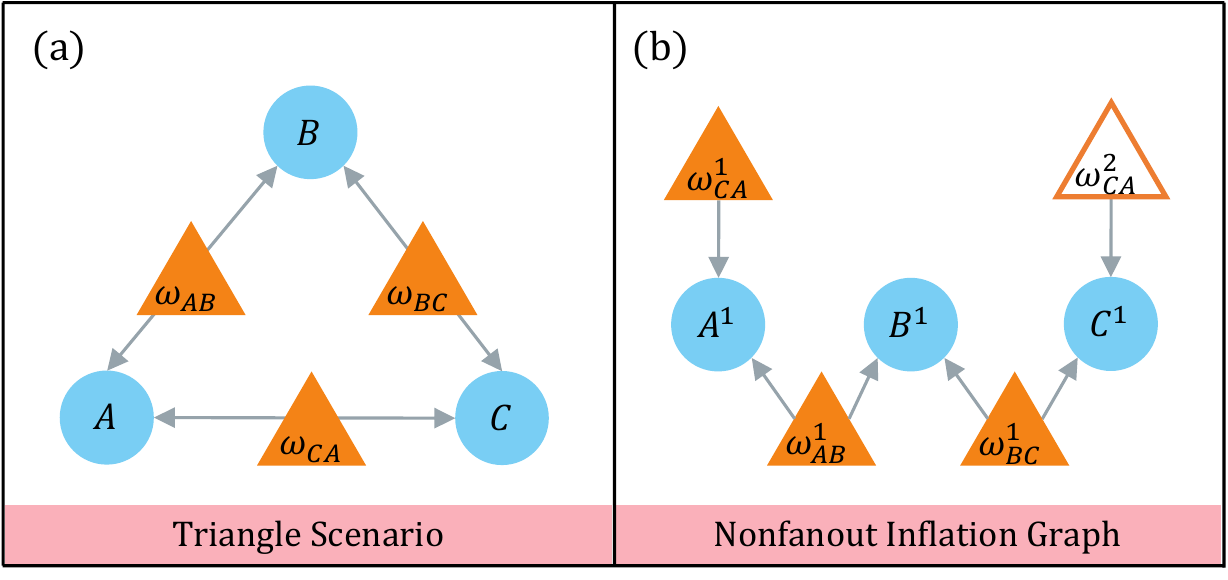}
\caption{Analysis for the triangle scenario.
(a) DAG of the triangle scenario, where $\GPT(G)\subsetneq\MI(G)$. The circles represent observable variables, including $A,B$, and $C$. The triangles represent the latent variables, including $\omega_{AB},\omega_{BC}$, and $\omega_{CA}$.
(b) A small nonfanout inflation of the triangle causal structure. $\{A^1,B^1\}$ is an injectable set, as is $\{B^1,C^1\}$. By contrast, the pair $\{A^1,C^1\}$ is not injectable, as the ancestors of $\{A^1,C^1\}$ include two different vertices with the same name, namely $\omega_{CA}^1$ and $\omega_{CA}^2$. We leave the added triangle unshaded to highlight that it has a different copy index from the others.}
\label{fig:Triangle}
\end{figure}

The proof by \citet{henson2014theory} can retrospectively be understood as a special example of a proof by \emph{nonfanout inflation}. The concept of nonfanout inflation was introduced by~\citet{wolfe2019inflation} as a restricted variant of a more general inflation technique. Here, we present a stand-alone mini guide to nonfanout inflation.

\begin{definition}[\textbf{Nonfanout Inflation Graphs}]\label{def:nonfanoutinflationgraph}
Consider a causal structure denoted by a DAG $G(V,W,\mathcal{E})$, where $V$ is the set of vertices representing observed variables, $W$ is the set of vertices representing latent variables, and $\mathcal{E}$ is the edge set. A graph $G'(V',W',\mathcal{E}')$ is said to be \textbf{a nonfanout inflation of $\boldsymbol G$} if the vertices in $G'$ are \emph{copies} of those in $G$, and if furthermore $G'$ \emph{locally} matches the structure of $G$.\footnote{Note that \citet{wolfe2019inflation} claims to provide two equivalent definitions of an inflation graph. Technically, the two definition proffered by \citet{wolfe2019inflation} are inequivalent. Our first condition in Definition~\ref{def:nonfanoutinflationgraph} corresponds to their second definition. The second condition in Definition~\ref{def:nonfanoutinflationgraph} restricts to \emph{nonfanout} inflations.} Formally, every vertex in $G'$ is indicated by both a \emph{name} and a \emph{copy index}. The name must correspond to a vertex name in $G$, and the copy index can be any integer. Define the deflation operation, $\Deflate$, as the quotient map that maps vertices of $G'$ to vertices in $G$ by dropping their copy indices. Then, $G'$ is a nonfanout inflation graph of $G$ if and only if:
\begin{enumerate}[leftmargin=*,nosep]
\item \textbf{\textup{Local Matching:}} For any vertex $u'$ in $G'$ (representing either an observed or latent variable), it holds that $\Deflate(\mathrm{pa}_{G'}(u'))=\mathrm{pa}_{G}(\Deflate(u'))$.
\item \textbf{\textup{Nonfanout:}} For any vertex $w'\in W'$ representing a \emph{latent} variable in $G'$, it holds that its child vertices, $\mathrm{ch}_{G'}(w')$, are \emph{duplicate-free} under the $\Deflate$ operation. That is, no two children of $w'$ can have the same name even if they have different copy indices.
\end{enumerate}
\end{definition}

We depict a nonfanout inflation graph with respect to the triangle scenario in Figure~\ref{fig:Triangle}(b). In the nonfanout inflation graph, we introduce another copy of the latent variable $\omega_{CA}$ in the original structure and label the two copies as $\omega_{CA}^{1}$ and $\omega_{CA}^{2}$, respectively. Note that, different from the visible-vertex-split graph and maximal interruption graph, there might be infinitely many nonfanout inflation graphs for a causal DAG.

Similarly to the construction of the causal models $\VVS(G)$ and $\MI(G)$, we shall define a causal model from the nonfanout inflation graphs via a projection operation. For this purpose, we first introduce the definition of \emph{Injectable Set}.

\begin{definition}[\textbf{Injectable Set}]
Consider a nonfanout inflation graph $G'(V',W',\mathcal{E}')$, where $V'$ is the set of vertices representing observed variables, $W'$ is the set of vertices representing latent variables, and $\mathcal{E}'$ is the edge set, and where each vertex has both a name and a copy index. 
A subset of the visible vertices $V^{\prime\star}\subseteq V'$ is said to comprise an injectable set within the vertices of $G'$, if the induced subgraph of $G'$ corresponding to $V^{\prime\star}$ and its ancestors gets mapped to the corresponding induced-subgraph of $G$ over $\Deflate(V^{\prime\star})$ and \emph{its} ancestors therein when the operation $\Deflate$ is applied to the entire subgraph. That is, the causal history of $V^{\prime\star}$ in $G'$ is identical to the causal history behind $\Deflate(V^{\prime\star})$ in the original (uninflated) graph $G$. 
\label{Def:Injectabl}
\end{definition}

The deflation operation is similar to the projection operations in defining the visible-vertex-splitting and maximal interruption models. We have a similar result as Propositions \ref{prop:projection} and \ref{prop:projectionMI}, which relates the correlations over the observable variables in a causal DAG to the correlations in its nonfanout inflation graphs.

\begin{proposition}\label{prop:nonfanoutinflation}
 Let $G(V,W,\mathcal{E})$ be a causal DAG, and let $G'(V',W',\mathcal{E}')$ be any nonfanout inflation of $G$. 
 Given a GPT $\mathcal{G}$, a distribution $p(X_V)$ is within $\mathcal{G}(G)$ only if there exists some distribution  $q(X_{V'})$ over the observed vertices of the nonfanout inflation graph $G'$ that is within $\mathcal{G}(G')$, such that for every injectable set $V^{\prime\star}$, the marginal distribution of $q$ agrees with the corresponding marginal distribution of $p$. That is, $p\left(X_{\Deflate(V^{\prime\star})}\right)=q(X_{V^{\prime\star}})$.
 \end{proposition}
 
 Of course, it must follow that for every GPT $\mathcal{G}$, if we recycle the causal components of the GPT model on $G$ to define a GPT causal model on the inflation graph $G'$, it will also hold that $q(X_{V'})\in \mathcal{G}(G')$ is in the set of ordinary Markov models for the inflation graph $G'$; and if $V^{\prime 1}$ and  $V^{\prime 2}$ are two (ordered) subsets of $V'$ that are causally copy-isomorphic to one another relative to $G'$, then the two marginal distributions of $q$ on the corresponding variables must agree. That is, $q(X_{V^{\prime 1}})=q(X_{V^{\prime 2}})$. Causally copy-isomorphic subgraphs are gedankenexperiments constructed by wiring the same types of devices in the same way, and thus they must realize the same marginal distributions. For brevity, we defer the formal definition of causally copy isomorphism to Appendix~\ref{app:defcopyisomorphism}.

Marginal distribution consistency over the variables of the injectable set as well as across the variables of causally copy-isomorphic subgraphs follows from positing \emph{some} physical realization of all the system states, transformations, and measurements. Essentially, if two vertex sets have indistinguishable causal pasts, they must give rise to the same distribution over their variables. Causally copy-isomorphisms naturally generalize the idea of an injectable set: whereas injectable sets relate vertices of the original graph to vertices of the inflation graph under ancestral closure and the $\Deflate$ map, copy-isomorphisms analogously relate pairs of vertex sets within the inflation graph to each other. See Ref.~\cite[Section V-D]{wolfe2019inflation} for an extended discussion relating these axioms of nonfanout inflation to consequences of a GPT realization. 

Proposition~\ref{prop:nonfanoutinflation} inspires us to define a GPT-agnostic model relative to any particular inflation graph. As a prerequisite, we introduce the notion of the inflation test.

\begin{definition}\label{defn:nonfanoutinflation}
 Let $G(V,W,\mathcal{E})$ be a causal DAG, and let $G'(V',W',\mathcal{E}')$ be any nonfanout inflation of $G$. A distribution $p(X_V)$ is said to \textbf{pass the inflation test} with respect to the nonfanout inflation graph $G'$ only if there exists some distribution $q(X_{V'})$ over the observed vertices of the nonfanout inflation graph $G'$ such that
 \begin{enumerate}[leftmargin=*,nosep]
\item For every injectable set $V^{\prime\star}$, the marginal distribution of $q$ agrees with the corresponding marginal distribution of $p$. That is, $p\left(X_{\Deflate(V^{\prime\star})}\right)=q(X_{V^{\prime\star}})$. 
\item $q$ is in the set of ordinary Markov models for $G'$, namely, that $q$ exhibits a conditional independence condition corresponding to each $d$-separation relation in $G'$ that pertains exclusively to vertices within $V'$. 
\item If $V^{\prime 1}$ and  $V^{\prime 2}$ are two (ordered) subsets of $V'$ that are causally copy-isomorphic to one another relative to $G'$, then the two marginal distributions of $q$ on the corresponding variables must agree. That is, $q(X_{V^{\prime 1}})=q(X_{V^{\prime 2}})$.
\end{enumerate}

\end{definition}

As an example of how the inflation test unveils constraints to physically realizable distributions in a causal DAG, we show why the distribution in Eq.~\eqref{eq:synchr} is not implementable by the triangle scenario. Consider the nanfanout inflation graph shown in Figure~\ref{fig:Triangle}(b). In this inflation graph, we observe the $d$-separation between $A^{1}$ and $C^{1}$, which should indicate the independence between the two random variables. However, the distribution in Eq.~\eqref{eq:synchr} suggests a perfectly synchronous correlation between them, leading to a contradiction.

Nonfanout inflation has been used to great effect by physicists in showing that large classes of multipartite entangled states cannot be realized by quantum processes limited to bipartite sources and restricted cross-site communication~\cite{CoiteuxRoy2021,Hansenne2022,Makuta2023,Wang_2024}. We encourage the reader to peruse those references for plentiful examples of nonfanout inflation.
For our discussion of investigating the most generally possible physical theory, we pair the idea of passing the nonfanout inflation test with the previous model based on maximal interruption, which brings us to the final model in this work.

\begin{definition}[The Maximal Interruption Plus Nonfanout Inflation Model]
   Given a causal DAG $G$, let $\MaxInt(G)$ be its maximal interruption graph, and let $G''$ be some nonfanout inflation of $\MaxInt(G)$.\footnote{To be clear: $G''$ is \emph{not} a nonfanout inflation of the original DAG $G$, but rather, $G''$ is a nonfanout inflation of the maximal interruption graph of $G$.} Now, a model for distributions on $G$ can be defined \emph{relative} to the specific nonfanout inflation graph $G''$ as $\Proj\bigl(q \text{ on }\MaxInt(G)\bigr)$, where $q$ is restricted to all distributions passing the inflation test with respect to \(G''\). 
   
   The \textbf{maximal interruption plus nonfanout inflation} model, $\MINF(G)$, is defined analogously, but instead of merely with respect to one particular $G''$, the distribution $q$ that ultimately projects to $p$ is required to pass the nonfanout inflation test with respect to every possible nonfanout inflation of $\MaxInt(G)$.
\end{definition}

By construction and Proposition~\ref{prop:nonfanoutinflation}, we have the following result.

\begin{corollary}
 Let $G$ be any causal DAG.
 Then, $\GPT(G)\subseteq\MINF(G)$.
\end{corollary}

In general, by resorting to the nonfanout inflation graphs of $\MaxInt(G)$, the model strengthens its constraints relative to merely requiring that $q$ be ordinary Markov with respect to $\MaxInt(G)$ as per Definition~\ref{def:MIset}. 

Admittedly, there are some graphs where $\MINF(G)=\MI(G)$. For instance, any graph with a single latent vertex is of that kind. 
For nonfanout inflation to imply constraints \emph{beyond} the ones implied by maximal interruption, one would have to construct at least one weakly-connected nonfanout inflation graph (of the maximal interruption graph) that is not ``trivial." An inflation graphs is trivial if just one injectable set comprises the entirety of the graph's observed vertices. In such a case, every distribution relative to the maximal interruption graph can be immediately recycled to pertain to the observed vertices of the trivial inflation graph. If the distribution is Markov with respect to the interruption graph, then its corresponding image on the trivial inflation graph will automatically pass the inflation test as per Definition~\ref{defn:nonfanoutinflation} by construction. Note that a trivial inflation graph cannot have any copy-isomorphic pairs of vertices: since the original graph has only one vertex for each name, the injectabilility of the entire set of observed vertices implies no pair of vertices with the same names but different copy indices. Therefore, DAGs where no nontrivial inflations can be constructed satisfy $\MINF(G)=\MI(G)$.

\section{Discussion}
In summary, we start with the study of the physical interpretation of the nested Markov model, an algebraic approach to causality characterization. For this purpose, we look into GPTs, where latent variables may represent generalized physical states subjected to certain probabilistic rules and relativity principles. 
We show that while GPT causal theories comply with the nested Markov model, there exist causal structures in which the nested Markov model allows ``non-physical'' probability distributions among observable variables.
To summarize our findings, we present the following theorem.
\begin{theorem}%
    For any causal structure denoted as DAG $G$, the following inclusions always hold:
    \begin{align}\label{eq:main}
    \begin{split}
        \C(G)&\subseteq\GPT(G)\\
        &\subseteq\MINF(G)\\
        &\subseteq\MI(G)\\
        &\subseteq\VVS(G)\\
        &\subseteq\NEST(G)\\
        &\subseteq\IND(G).
    \end{split}
    \end{align}
    Moreover, there exists causal structures illustrating the strictness of the inclusions for all except $\GPT(G)\subseteq\MINF(G)$.
    \label{thm:main}
\end{theorem}

In Table~\ref{table:constraints}, we compare the constraints in the causal models $\NEST(G)$, $\VVS(G)$, $\MI(G)$, and $\MINF(G)$. Indeed, the Zariski closures of these models all coincide, and they completely characterize all the equality constraints on the GPT-realizable distributions.
Nevertheless, the models differ in inequality constraints. In the first level, the nested Markov model $\NEST(G)$ does not specify any inequality constraint and is thus the loosest set among these models. To define the models $\VVS(G)$, $\MI(G)$, and $\MINF(G)$, we first construct extended hypergraphs and then apply particular projection principles. The projection principles defining these models are successively becoming more restrictive, leading to more types of non-trivial inequality constraints and approaching the accurate characterization of ``physical requirements'' for an observable distribution in a causal structure. Unlike $\C$ or $\GPT$, all the models proposed in this paper are defined by statistically testable constraints derived from the graph structure, following a specific constraint-generating prescription. The $\MINF$ model is the tightest model we were able to envision that is defined by such constraints, while still admitting \emph{all} OPT-compatible distributions.

\begin{table}[hbt!]
\centering
    \begin{tabular}{c|cccc}
    \hline
    \hline
       \diagbox{\thead{\footnotesize Constraints}}{\thead{\footnotesize Model}} & {\footnotesize $\NEST(G)$} & {\footnotesize $\VVS(G)$} & {\footnotesize $\MI(G)$} & {\footnotesize $\MINF(G)$} \\
       \hline
       {\footnotesize equalities} & $\checkmark$ & $\checkmark$ & $\checkmark$ & $\checkmark$  \\
       \hline
       {\footnotesize e-separation inequalities}
       & $\cross$ & $\checkmark$ & $\checkmark$ & $\checkmark$ \\
       \hline
       {\footnotesize monogamy of nonlocality}
       & $\cross$ & $\cross$ & $\checkmark$ & $\checkmark$  \\
       \hline
       {\footnotesize no shared bit for triangle}
       & $\cross$ & $\cross$ & $\cross$ & $\checkmark$  \\
       \hline
       \hline
    \end{tabular}
    \caption{Comparison between different causal models considered in this work. The four causal models, $\NEST(G),\VVS(G),\MI(G)$, and $\MINF(G)$, successively give tighter characterizations of the set of valid distributions for a causal DAG. In the table, we tick the constraints that each model can recover.}
    \label{table:constraints}
\end{table}

A remaining open problem is whether there exist any distributions in $\MINF(G)\backslash\GPT(G)$. Nevertheless, we hope that the construction of $\MINF(G)$ can help clarify the notion of causality and improve our understanding of the origin of the gap between algebraically and physically motivated causal models.

The reader may wonder if it is possible to illustrate multiple strict model separations by considering a single DAG. Indeed, this is is the case. There are plenty of DAGs without \emph{any} ordinary or nested Markov constraints where nevertheless separations have been established between classical distributions versus those admitting a GPT realization. The most established classical-GPT gap of all time is that of the CHSH Bell test of Figure~\ref{fig:Bell}(a). That DAG, however, can only illustrate the single separation $\C(G)\subsetneq\GPT(G)$, as Bell scenarios have the relatively \emph{exceptional} property that $\GPT(G)=\IND(G)$.
Consider, however, the classical-GPT gap in the Instrumental scenario of Figure~\ref{fig:Bell}(b) as established by~\citet{VanHimbeeck2019}. That scenario, therefore, simultaneously establishes $\C(G)\subsetneq\VVS(G)\subsetneq\NEST(G)=\IND(G)$. The same separations can be illustrated while restricting all variables to have binary cardinality in the \emph{Unrelated Confounders} scenario, see Refs.~\cite[Eq. (38)]{lauand2023witnessing} and~\cite{lauand2024quantum}. We depict its causal DAG and associated visible-vertex-split graph in Figure~\ref{fig:UCScenario}. The existence of nonclassical correlations in the triangle scenario of Figure~\ref{fig:Triangle}(a) was established by~\citet{fritz2012beyond}. Since there is nothing to interrupt in the triangle scenario, it follows that $\C(G)\subsetneq\MINF(G)\subsetneq\MI(G)=\IND(G)$ is evident there.
Besides these causal DAGs, we provide our own example of two separations, $\C(G)\subsetneq\GPT(G)\subsetneq\NEST(G)$, which is inspired by the Guess-Your-Neighbor's-Input nonlocal game \cite{almeida2010guess}. We show the causal DAG and prove the separations in Appendix~\ref{app:GNYI}.

\begin{figure}[hbt!]
 \centering \includegraphics[width=\columnwidth]{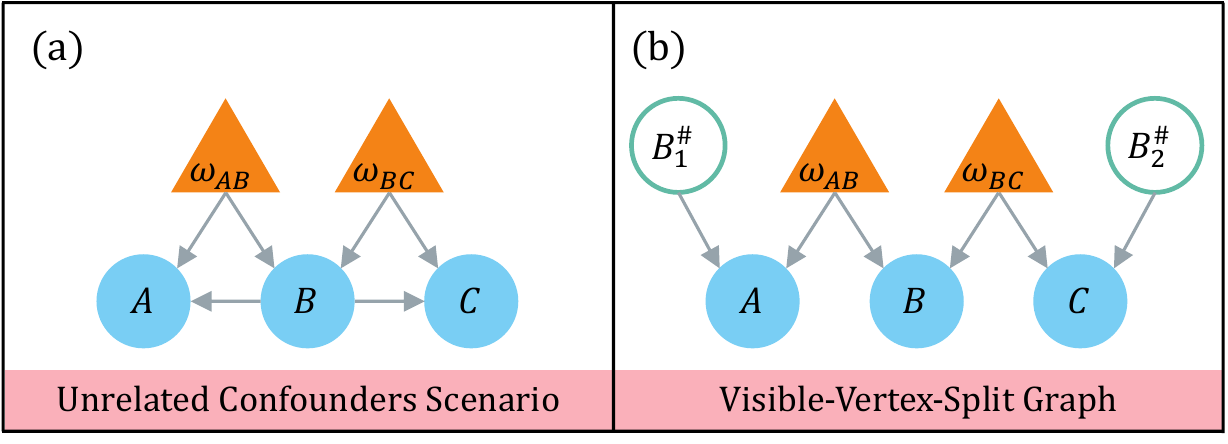}
\caption{Analysis for the unrelated confounders scenario. (a) The causal structure of the scenario. The circles represent observable variables, including $A,B$, and $C$, and the triangles represent the latent variables, $\omega_{AB}$ and $\omega_{BC}$. 
(b) The visible-vertex-split graph of (a). The green hollow circles are added vertices in the hypergraph, and a data post-selection of $B_1^{\#}{=}B_2^{\#}{=}B$ projects the structure back to the original scenario.}
\label{fig:UCScenario}
\end{figure}

\appendix
\section{Preliminaries}
\subsection{Generalized causality descriptions}\label{app:DAG}
In classical causal theories formalized as Bayesian networks, all the causal influences are modelled as conditional probabilities between random variables. In a GPT theory as per the assumptions in Box~\ref{box:GPTAssumption}, causal influences between solely observed variables are the same as Bayesian networks. Differently, the causal influence from a latent variable on an observed variable should be understood as the effect of a measurement, and the causal influence from a latent variable on another latent variable should be understood as a state transformation.
Here, we take the quantum theory as an example and describe the physical interpretation of causal DAGs.
In quantum theory, a general state transformation can be characterized as a quantum instrument~\cite{davies1970operational,ozawa1984quantum}.
In the description of quantum causal influences, first consider a simple case where one observable variable, $X$, is influenced by one latent variable. The latent variable provides an input state, $\rho$, to the quantum instrument. The quantum instrument takes a measurement on $\rho$, which is characterized as a POVM $\{M_x\}_x$, and outputs the measurement outcome $X$. The causal influence from $\rho$ to $X$ is reflected from the outcome probability, where $X=x$ is output with probability $p(x)=\tr(\rho M_x)$. With respect to the measurement outcome $x$, the quantum instrument also outputs a post-measurement system in the state $E_x\rho E_x^\dag/\tr(\rho E_x^\dag E_x)$, where $E_x$ is called a Kraus operator, satisfying $E_x^\dag E_x=M_x$ and $\sum_x M_x=I$.
Next, if the observed variable $X$ is also influenced by another observed variable $A$, $A$ acts as a classical control, such that $p(X{=}x|A{=}a)=\tr(\rho M^a_x)$, where each value of $A{=}a$ corresponds to a quantum instrument with POVM $\{M^a_x\}_x$.
Finally, in the case of a state transformation from a latent variable $\rho$ to another latent variable, the quantum instrument is a completely positive and trace-preserving map, which maps $\rho$ to another quantum state in a state $\sum_xE_x\rho E_x^\dag/\tr(\rho E_x^\dag E_x)$ for some set of Kraus operators.

\subsection{Operational-probabilistic theory}
Besides describing the observable events and probabilities, a GPT needs to explicitly specify the sets of valid physical states and state transformations. For the discussion of the nested Markov model, it is helpful to consider a more general framework, the operational-probabilistic theory (OPT)~\cite{chiribella2010probabilistic}, which does not require an explicit description of these elements. Some of our proofs will be based on the OPT framework. Note that every GPT is an OPT; hence results that hold for a general OPT are naturally valid for a general GPT. Nevertheless, it is an open question whether every OPT can be explained by a GPT.

We first review the OPT framework.
The first part of the framework is the definition of operations. A general physical operation is called a ``test,'' which generalizes the concept of quantum instrument in quantum information. The input system of the test carries a physical state in the OPT. After measuring the state, the test outputs a classical variable as the measurement outcome. In addition, the test also outputs an associated post-measurement state of the system. After specifying the operations, the framework describes the measurement outcomes in a probabilistic manner.
In general, the measurement outcome is a random variable, corresponding to the concept of ``events'' in probability theory. When two tests lead to the same probability distribution over the same set of events and the replacement of one test with another does not affect the input and output systems, they are said to have the same ``effect.'' If the measurement outcome can take only one possible value, the test is called deterministic. In the OPTs that we consider, we assume the deterministic effect to be unique. This assumption is equivalent to the assumption of \emph{no-signalling without interaction} of Box~\ref{box:GPTAssumption} \cite{chiribella2010probabilistic}.

To apply the OPT framework to the causal description in terms of the observed probability distributions in causal DAGs, we need to properly interpret the physical meanings of vertices and edges in the DAG. In particular, we hope the description is consistent with the $d$-separation criteria when the causal DAG is considered as a Bayesian network. For this purpose, we utilize the interpretation introduced by \citet{henson2014theory}. Under this interpretation, each vertex in the DAG, together with its incoming and outgoing edges, is described as a test, where the vertex loads the measurement outcome, the incoming edges load the input systems, and the outgoing edges load the output systems. If a vertex corresponds to an observed variable, we take it as a test with a non-trivial measurement outcome but a trivial output system; namely, only the measurement outcome is recorded, and the post-measurement system is discarded. If a vertex corresponds to a latent variable, we take it as a test with an output system in an OPT state but a trivial measurement outcome. Following the same notation convention as~\citet{henson2014theory}, we assign a fixed value $\Lambda\equiv1$ to a latent variable $\Lambda$ for completeness; hence every vertex in a causal DAG is assigned with a classical variable.

For vertex $v$, denote its parental vertices as $\mathrm{pa}(v)$, the joint system of its incoming edges as $\mathrm{inc}(v)$, and the joint system of its outgoing edges as $\mathrm{out}(v)$. To describe the causal influences from $\mathrm{pa}(v)$ on $v$, we record both the variable values of the vertices and the physical systems in a function
\begin{equation}
    \mathcal{T}\left(X_v|X_{\mathrm{pa}(v)}\right)^{\mathrm{inc}(v)}_{\mathrm{out}(v)}:\mathcal{X}_v\times\mathcal{X}_{\mathrm{pa}(v)}\rightarrow\mathbb{R},
\end{equation}
where $X_v$ and $X_{\mathrm{pa}(v)}$ represent the associated variables of $v$ and $\mathrm{pa}(v)$, $\mathcal{X}_v$ and $\mathcal{X}_{\mathrm{pa}(v)}$ represent their event spaces, and $\mathrm{inc}(v)$ and $\mathrm{out}(v)$ represent the incoming and outgoing systems of the test, respectively. Note that $X_{\mathrm{pa}(v)}$ might be a joint random variable that corresponds to multiple vertices.
In addition, if $v$ is a root vertex such that $\mathrm{pa}(v)=\emptyset,\mathrm{inc}(v)=\emptyset$, we may write the function simply as $\mathcal{T}\left(X_v\right)_{\mathrm{out}(v)}$.
According to the OPT formulation, the above function can be taken as a probability kernel over $\mathcal{X}_v$ given $\mathcal{X}_{\mathrm{pa}(v)}$.

Note that the kernel function calculation depends on the underlying input system $\mathrm{inc}(v)$ and the output system $\mathrm{out}(v)$. If the incoming or the outgoing system is trivial, we omit the corresponding notation. If the parental and child vertices both correspond to observed variables, $\mathrm{inc}(v)$ and $\mathrm{out}(v)$ become trivial systems, and $\mathcal{T}\left(X_v|X_{\mathrm{pa}(v)}\right)^{\mathrm{inc}(v)}_{\mathrm{out}(v)}$ degenerates to the usual conditional probability of $X_v$ given $X_{\mathrm{pa}(v)}$. For simplicity, we use the usual probability notations in this case. In addition, if a vertex is childless, we can determine the marginal distribution for the variables of other vertices by replacing this childless vertex with the deterministic effect. For a given probability kernel, the uniqueness of the deterministic effect guarantees the uniqueness of the marginal distribution~\cite{chiribella2010probabilistic,henson2014theory}. In later discussions, we sometimes use names of vertices and their corresponding variables interchangeably for brevity.

\subsection{Markov condition and $d$-separation}
As shown by \citet{hensen2015loophole}, we can generalize the Markov condition in a classical Bayesian network to a causal DAG interpreted by an OPT. 
Given a causal DAG $G$ with $m$ vertices, $v_1,\cdots,v_m$, where the first $n$ vertices $v_1,\cdots,v_n$ represent observed variables $X_{v_1},\cdots,X_{v_n}$.
A probability kernel $p$ over the observed variables is said to be generalized Markov with respect to $G$ in an OPT if there exists a decomposition of
\begin{equation}
    p(X_{v_1}{=}x_{v_1},\cdots,X_{v_n}{=}x_{v_n}) = \prod_{i{=}1}^{m}\mathcal{T}\left(x_{v_i}|x_{\mathrm{pa}_G(v_i)}\right)^{\mathrm{inc}(v_i)}_{\mathrm{out}(v_i)},
\end{equation}
where the function $\mathcal{T}$ is defined by the physical rules in the OPT.
Remember that we take the convention of $X_{v_i}\equiv1$ for $i=n{+}1,\cdots,m$.
For instance, in Figure~\ref{fig:Verma}(a), we have
\begin{equation}\label{eq:MadiationDist}
\begin{split}
    &p(X{=}x,A{=}a,B{=}b,C{=}c,\Lambda{=}1) \\ =&p(x)\mathcal{T}(a|x)^{\mathcal{H}_A}p(b|a)\mathcal{T}(c|b)^{\mathcal{H}_C}\mathcal{T}(\Lambda=1)_{\mathcal{H}_{AC}}.
\end{split}
\end{equation}
Note that the term $\mathcal{T}(a|x)^{\mathcal{H}_A},\mathcal{T}(c|b)^{\mathcal{H}_C},\mathcal{T}(\Lambda=1)_{\mathcal{H}_{AC}}$ are not independent of each other, as is reflected from the dependence on the non-trivial incoming or outgoing systems. Here, the trivial classical variable $\Lambda\equiv1$ is assigned to the vertex of the latent variable, and we borrow the notation for Hilbert spaces in quantum theory to denote the systems, which are $\mathcal{H}_{AC}$ for the output system of $\Lambda$, $\mathcal{H}_A$ for the input system of $A$, and $\mathcal{H}_C$ for the input system of $C$. The probability kernel decomposition captures the essence of the Markov condition, where the probability distribution of each vertex is only directly influenced by its parental vertices. Thus, we say $p(X,A,B,C,\Lambda)$ is generalized Markov.

In classical Bayesian networks, the conditional independences among observed variables are equivalent to the graphical criteria of $d$-separation~\cite{pearl1995causal}. Namely, two observed variables $X_{v_i}$ and $X_{v_j}$ are independent conditioned on a set of variables $X_{v_{k_1}},\cdots,X_{v_{k_m}}$ if and only if the corresponding vertices of $v_i$ and $v_j$ are $d$-separated by the set of vertices $v_{k_1},\cdots,v_{k_m}$. For convenience of later discussions, we briefly review the $d$-separation criteria here. Consider a DAG with vertices $V = \{v_1,\cdots,v_n\}$. We say that two vertices $v_i, v_j$ are $d$-separated by a set of vertices $W \subseteq V \setminus \{v_i, v_j\}$ ($W$ can be an empty set) if for any path $v_i =: v_{k_1}, \ldots, v_{k_m} := v_j$ that links $v_i$ and $v_j$ in $G$, it holds that
\begin{enumerate}[leftmargin=*,nosep]
    \item If there are ``colliders'' in the path, namely, the kind of vertex $v_{k_{m'}}$ such that both $v_{k_{m' - 1}}$ and $v_{k_{m' + 1}}$ are connected to $v_{k_{m'}}$ through an edge pointing towards $v_{k_{m'}}$, i.e., ${v_{k_{m' - 1}} \to v_{k_{m'}} \leftarrow v_{k_{m' + 1}}}$, it holds that there exists at least one collider that both the collider itself and its descendents do not belong to $W$;
    \item Denoting $V'$ as the set of vertices that are non-colliders in the path other than $v_{k_1}$ and $v_{k_m}$, then we have that either $V' = \emptyset$ or $V' \cap W \neq \emptyset$.
\end{enumerate}

Furthermore, we can generalize the $d$-separation criteria to sets of vertices. Vertex sets $A$ and $B$ are $d$-separated by the vertex set $C$ if for any pair of vertices $v\in A$ and $u\in B$, they are $d$-separated by the vertex set $C$.

A remarkable property of the generalized Markov condition is that it preserves the equivalence between (conditional) independences and (conditional) $d$-separation in OPTs~\cite{henson2014theory}.
For instance, in Figure~\ref{fig:Verma}(a), as the causal influence of $X$ on $B$ is mediated by $A$, we have the independence condition of $X\independent B|A$ shown by the following kernel decomposition:
\begin{equation}
\begin{split}
    p(x,a,b)&=\sum_{C=c}p(X{=}x,A{=}a,B{=}b,C{=}c) \\
    &=
    p(x)\mathcal{T}(a|x)^{\mathcal{H}_A}p(b|a).
\end{split}
\end{equation}
The marginalization over $C$ is valid since it is childless in the DAG. Graphically, the conditional independence can be shown by blocking the node $A$ in the original directed path of $X\rightarrow A\rightarrow B$, the conditional $d$-separation criterion.

\subsection{Nested Markov property}\label{app:NestedDef}
In this section, we review the definition of the nested Markov property. As shown by the example in Sec.~\ref{Sec:Nested}, we take a graph-based approach to the definition. Consider a causal DAG $G(V,W,\mathcal{E})$, where $V$ is the set of vertices representing observed variables, $W$ is the set of vertices representing latent variables, and $\mathcal{E}$ is the edge set. 
Depending on the probability queries used in the analysis, we may take a set of vertices in $V$, $F\subseteq V$, as fixed, which correspond to fixing the values of the variables they represent. The set of vertices $F$ can be chosen freely; nevertheless, such vertices must be parentless. The other vertices, $R=V\backslash F$, are called random.
Following \citet{tian2002general,evans2017margins}, we decompose $R$ into districts. First, we say two vertices $v_i,v_j\in R$ are linked by a bi-directed path, $v_i\leftarrow w_k\cdots w_l\rightarrow v_j$, such that $w_k,w_l\in W$ and there is a path between $w_k$ and $w_l$ going strictly within $W$. Then, a district, $D$, is defined as the maximal inclusion set of vertices in $R$ representing observed variables that are linked by bi-directed paths. 
Given district $D$, its associated subgraph $G[D]$ is defined as a causal DAG, of which 
\begin{itemize}[leftmargin=*,parsep=0pt]
    \item The set of vertices representing observed variables is given by $D$ and $\mathrm{pa}(D)\backslash D\cap V$, with the vertices in $D$ taken as random and the vertices in $\mathrm{pa}(D)\backslash D\cap V$  taken as fixed;
    \item The set of vertices representing latent variables is given by the maximal inclusion set of vertices in $W$ that is on a bi-directed path linking two vertices in $D$;
    \item The edge set contains the directed edges $u\rightarrow v$, where $v\in D$ and $u\in\mathrm{pa}(D)\backslash D\cap V$, and all the directed edges on the bi-directed paths. 
\end{itemize}
Based on these notions, we are ready to define the nested Markov property.

\begin{definition}[Nested Markov property, Definition~3.1 in~\citet{evans2017margins}]
    Given a causal DAG $G(V,W,\mathcal{E})$ with $V=R\cup F,R\cap F=\emptyset$, where $R$ and $F$ are the sets of vertices that are random and fixed, respectively, suppose a distribution $p$ defines a valid probability kernel over the state space of $R$, $\mathcal{X}_{R}$, indexed by $\mathcal{X}_{F}$. Then, it obeys the nested Markov property with respect to $G(V,W,\mathcal{E})$ if $R=\emptyset$, or both

\begin{enumerate}[leftmargin=*]
  \item $p$ can be factorized over districts $D_1,\cdots,D_l$ of $G$:
    \begin{equation}\label{eq:districtdecomp}
  p(x_R|x_F)=\prod_{i=1}^{l}g_i(x_{D_i}|x_{\mathrm{pa}(D_i)\backslash D_i\cap V}),
    \end{equation}
where every $g_i$ is a probability kernel obeying the nested Markov property with respect to the subgraph $G[D_i]$ if $l\geq 2$ or $F\backslash\mathrm{pa}_{G}(R)\neq\emptyset$,
  
  \item $\forall v\in R$ such that $\mathrm{ch}_{G}(v)=\emptyset$, the marginal kernel over the state space of $R\backslash v$ indexed by $\mathcal{X}_F$,
    \begin{equation}\label{eq:vermamarginal}
      p(x_{R\backslash v}|x_F)=\sum_{x_v}p(x_R|x_F),
    \end{equation}
    obeys the nested Markov property with respect to the subgraph $G[R\backslash\{v\}]$.
\end{enumerate}
\label{def:nested}
\end{definition}

As can be seen from the graphical approach to decomposing the probability distribution with respect to its associated causal DAG, in brief, a probability distribution having the nested Markov property is such a distribution that its identifiable probability kernels satisfy the following property: a $d$-separation condition among vertices representing observed variables implies conditional-independence among these observed variables.

We note that Eq.~\eqref{eq:vermamarginal} defines the nested Markov property in a recursive manner: After the marginalization over a childless free vertex $v$ in Eq.~\eqref{eq:vermamarginal}, for the resulted subgraph and its associated probability kernel, repeat the two steps in Eq.~\eqref{eq:districtdecomp} and~\eqref{eq:vermamarginal}. Particularly, a DAG might contain multiple childless free vertices, and we need to marginalize each of them and apply the same procedure to each resulted subgraph separately. This definition naturally defines a procedure to search for nested Markov constraints similar to the Verma constraint, which ends for a subgraph if it only contains a single random vertex. 
In the marginalization step of Eq.~\eqref{eq:vermamarginal}, if we find the marginal kernel, $\sum_{x_v}p(x_R|x_F)$, is independent of the variable values of a strict subset of the fixed vertices, $F'\subsetneq F$, then we find a nested Markov constraint similar to Eq.~\eqref{eq:Verma}.

\section{Proof of Proposition~\ref{prop:already_seminetwork}}\label{app:prop1}
Proposition~\ref{prop:already_seminetwork} provides a prerequisite for the existence of a non-trivial nested Markov constraint. The contraposition of Proposition~\ref{prop:already_seminetwork} is that for a DAG $G$ to have non-trivial nested Markov constraints, a necessary condition is the existence of a triplet of vertices that are strictly topologically ordered, ${x\rightarrow v\rightarrow y}$, with ${v\in V}$ representing an observed variable.

To prove the above statement, we inspect the definition of the nested Markov property in Definition~\ref{def:nested}. Suppose there is a non-trivial nested Markov constraint in a causal DAG $G$. Then, in the probability kernel factorization in Eq.~\eqref{eq:districtdecomp}, there exists a district $D_i$ such that the kernel $g_i(x_{D_i}|x_{\mathrm{pa}(D_i)\backslash D_i\cap V})$ cannot be specified by a usual Markov constraint in the original DAG $G$. That is, there exists vertices $v,u\in D_i$, such that one of the child vertices of $v$, $y\in\mathrm{ch}(v)$, belongs to $\mathrm{pa}(u)\backslash D_i\cap V$; and there exists a parent vertex of $v$, $x\in\mathrm{pa}(v)$, belonging to $\mathrm{pa}(v)\backslash D_i\cap V$. Nevertheless, this indicates a strictly topological order of $x\rightarrow v\rightarrow y$, and all these three vertices represent observed variables. This fact contradicts the definition of (semi-) network-type DAGs. Therefore, Proposition~\ref{prop:already_seminetwork} is proved.

\begin{figure}[hbt!]
\centering 
 \includegraphics[width=0.65\columnwidth]{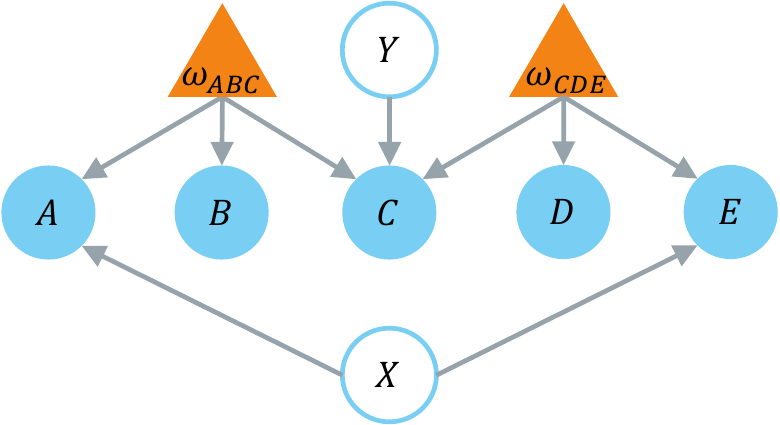}
\caption{An example of a semi-network-type causal DAG. The DAG contains vertices representing observed variables, $X,Y,A,B,C,D,E$, and vertices representing latent variables, $\omega_{ABC},\omega_{CDE}$. Taking all the vertices representing observed variables as free, the free vertices can be divided into three districts: $\{X\},\{Y\}$, and $\{A,B,C,D,E\}$.}
\label{fig:example}
\end{figure}

\section{Markov property in GPTs}\label{app:nested}
To understand the proof of Theorem~\ref{thm:nested}, which states that the nested Markov constraints hold valid in all GPTs, we will actually prove only the various lemmas and propositions from which Theorem~\ref{thm:nested} then trivially follows. In particular, we focus on proving Proposition~\ref{prop:projection}, from which Lemma~\ref{lemma:easy_part} almost immediately follows, and on proving Lemma~\ref{lemma:hard_part}. Theorem~\ref{thm:nested} is an immediate consequence of considering Lemmas~\ref{lemma:easy_part}~and~\ref{lemma:hard_part} together.

\subsection{Proof of Propositions~\ref{prop:projection} and \ref{prop:projectionMI}}\label{app:prop2}

Consider a causal DAG $G(V,W,\mathcal{E})$, where $V$ is the set of vertices representing observed variables, $W$ is the set of vertices representing latent variables, and $\mathcal{E}$ is the edge set.
Its associated visible-vertex-split graph, $\VVSplit(G)(V\cup V^{\#},W,\mathcal{E}')$, is defined as in Definition~\ref{def:projection}, and its associated maximal-interruption graph, $\MaxInt(G)(V\cup U,W,\mathcal{E}')$, is defined as in Definition~\ref{def:hypergraph}.
In addition, consider a specific GPT $\mathcal{G}$ that satisfies the assumptions in Box~\ref{box:GPTAssumption}.
By definition, we can prove Propositions~\ref{prop:projection} and \ref{prop:projectionMI} concurrently by showing the \emph{if} direction of Proposition~\ref{prop:projection} and proving the \emph{only if} direction of Proposition~\ref{prop:projectionMI}.

\subsubsection{Sufficient condition in Proposition~\ref{prop:projection}}
We first show the \emph{if} direction of Proposition~\ref{prop:projection}, which almost follows directly by definition. Suppose a probability kernel $q(X_V|X_{V^{\#}})$ among observed variables of $V\cup V^{\#}$ in $\VVSplit(G)$ can be generated within $\mathcal{G}$. Note that $\VVSplit(G)$ is a semi-network causal DAG, where every vertex representing an observed variable is either a root or an end. With respect to the underlying GPT $\mathcal{G}$, a root observable maps to a test with trivial input and output systems, whose value controls the test to be performed at the child vertices; an end observed variable maps to a test with a trivial output system, whose value corresponds to the outcome of the measurement on its input system. Now, consider a vertex $v\in V_{\mathrm{split}}$, which represents an observed variable being split, and its associated vertices $\{v_i^{\#}\in V^{\#}\}_i$. In $\VVSplit(G)$, the vertex $v$ is an end vertex, and the vertices $\{v_i^{\#}\in V^{\#}\}_i$ are root vertices. Consequently, the output system of the test that $v$ is mapped to trivially matches the input systems of the tests that $v_i^{\#}$ are mapped to. Therefore, it is physically valid to wire these vertices by plugging the output system of $v$ to the input systems of $v_i^{\#}$ and fix the probability kernel to be $X_{v_i^{\#}}\equiv X_v$. The wiring operation is exactly the projection operation defined in Definition~\ref{def:projection}. Since the operations are physically valid within $\mathcal{G}$, the projected distribution $p(X_V)=\Proj(q(X_V|X_{V^{\#}}))$ is thus a valid distribution.

\subsubsection{Necessary condition in Proposition~\ref{prop:projectionMI}}
Next, we prove the \emph{only if} direction of Proposition~\ref{prop:projectionMI}. That is, given a distribution $p(X_V)$ over the observed variables in $G$ that is admitted by a GPT $\mathcal{G}$, it can be extended to a valid probability kernel $p(X_V|X_U)$ over the observed variables in $\MaxInt(G)$ admitted by $\mathcal{G}$. In constructing the maximal interruption graph,
\begin{enumerate}[leftmargin=*,parsep=0pt]
    \item Suppose $v\in V_{\mathrm{split}}$, which represents an observed variable and has both incoming and outgoing edges. With respect to the underlying GPT $\mathcal{G}$, the test that $v$ is mapped to has a trivial output system, and the random variable represented by $v$, $X_v$, serves as a classical control on the tests to be performed at its child vertices. Hence for an edge $(e:v\rightarrow y)$ outgoing from $v$, after replacing it by an edge originating from a newly-added vertex, $u_{e}\in U$, which is equipped with a classical variable $X_{u_{e}}$ sharing the same sample space as $X_v$, there is not a change in the input system of the test at vertex $y$ in $\MaxInt(G)$. Additionally, the classical control of the test at vertex $y$ in $\MaxInt(G)$ is isomorphic to the original one in $G$. Therefore, we can extend the original $X_v$-controlled test at $y$ in $G$ to the $X_{u_{e}}$-controlled test at $y$ in $\MaxInt(G)$ in the same GPT $\mathcal{G}$. As the output system of $v$ is trivial, there is not a change in the test at $v$.
    
    \item Suppose $v\in V_{\mathrm{rip}}\backslash V_{\mathrm{split}}$, which is a root vertex with multiple child vertices, $v\rightarrow\{y_j\}_j$. Following a similar argument as above, there is not a change in the input system of the test at $y_j$ when the edge $(e_j:v\rightarrow y_j)$ is replaced by an edge that originates from a newly-added vertex, $u_{e_j}\in U$, equipped with a classical variable $X_{u_{e_j}}$ sharing the same sample space as $X_v$. Additionally, the classical control of the test at vertex $y_j$ in $\MaxInt(G)$ is isomorphic to the original one in $G$. Therefore, we can extend the original $X_v$-controlled test at $y_j$ in $G$ to the $X_{u_{e_j}}$-controlled test at $y_j$ in $\MaxInt(G)$ in the same GPT $\mathcal{G}$. As the output system of $v$ is trivial, there is not a change in the test at $v$.
\end{enumerate}

The above arguments establish a valid probability kernel over observed variables in $\MaxInt(G)$, $q(X_V|X_U)$. Then, following the same argument that we show the \emph{if} direction of Proposition~\ref{prop:projection}, we have $\Proj(q(X_V|X_U))=p(X_V)$ in the original graph. This finishes the proof.

\subsection{Proof of Lemma~\ref{lemma:hard_part}}\label{app:lemma1}
\begin{proof}
 Consider a causal DAG $G(V,W,\mathcal{E})$, where $V$ is the set of vertices representing observed variables, $W$ is the set of vertices representing latent variables, and $\mathcal{E}$ is the edge set. We take all the vertices in $V$ as random vertices. In addition, we denote the visible-vertex-split graph of $G$ as $\VVSplit(G)(V\cup V^{\#},W,\mathcal{E}')$, which is obtained according to Definition~\ref{def:SWIG}, where vertices in $V^{\#}$ are added as per the subset of vertices to be split, $V_{\mathrm{split}}\subseteq V$. When denoting the vertices and the variables they represent, we use the same notations as those in the proof of Proposition~\ref{prop:projection} in Appendix~\ref{app:prop2}.
 
 To show $\VVS(G)\subseteq\NEST(G)$, we explain that for every probability distribution $p\in\IND(\VVSplit(G))$, the projected distribution $\Proj(p)$ admits a factorization that respects the district decomposition in $G$ and enjoys the nested Markov property with respect to $G$. For this purpose, we inspect the correspondence between the subgraphs defined by the districts in $G$ and $\VVSplit(G)$ and consequently their associated distributions.
\begin{enumerate}[leftmargin=*]
    \item According to the construction of $\VVSplit(G)$, 
    we do not modify the bi-directed paths in $G$.
    In other words, suppose the vertex set $V$ of $G$ can be decomposed into districts $\{D_i\}_{i=1}^{t}$. Then, the set of vertices in $D_i$ also defines a district in $\VVSplit(G)$. As $\bigcup_{i} D_i=V$, the district partition for $\VVSplit(G)$ within the vertex set $V$ is the same for $G$. In $\VVSplit(G)$, all the newly-added vertices, namely those in the vertex set $V^{\#}$, are root vertices, each of which has one single child vertex. Therefore, each added vertex $v_{i}^{\#}$ defines a single-vertex district in $\VVSplit(G)$.
    
    \item Consider a probability distribution $p\in\IND(\VVSplit(G))=\NEST(\VVSplit(G))$. Then, according to the nested Markov property in Definition~\ref{def:nested}, this distribution has a factorization with respect to the subgraphs:
    \begin{equation}   p=\prod_{i=1}^{t}g_i(x_{D_i}|x_{\mathrm{pa}_{\VVSplit(G)}(D_i)\backslash D_i})p(x_{V^{\#}}),
    \end{equation}
    where for the subgraph $\VVSplit(G)[D_i]$, the conditional probability distribution is given by $g_i(x_{D_i}|x_{\mathrm{pa}_{\VVSplit(G)}(D_i)\backslash D_i})$, and we abbreviate $\prod_{v_k^{\#}\in V^{\#}}p(x_{v_k^{\#}})$ as $p(x_{V^{\#}})$. Furthermore, $g_i$ obeys the nested Markov property with respect to the subgraph $\VVSplit(G)[D_i]$. After the projection, the values of variables in $V^{\#}$ are fixed as $x_{v_k^{\#}}\equiv x_{v_k}$, where $v_k^{\#}$ is an added vertex as per $v_k\in V_{\mathrm{split}}$; hence the probability distribution $p(x_{V^{\#}})$ is fixed. 
    We simply need to check the projected distribution $\Proj(g_i)$ of each district for the examination of the nested Markov property. 
    
    We argue that $\Proj(g_i)$ obeys the nested Markov property with respect to the subgraph $G[D_i]$. We simply need to focus on the vertices in $\mathrm{pa}_{\VVSplit(G)}(D_i)\backslash D_i\cap V^{\#}$ for this statement. Consider a general vertex of such, $v^{\#}$, which is newly added as per $v\in V_{\mathrm{split}}$:
    \begin{enumerate}[leftmargin=*]
        \item If $v\notin D_i$, then the projection operation, i.e., ${x_{v^{\#}}\equiv x_{v}}$, does not change the conditional distribution $g_i$. Its effect is merely a renaming of the term being conditioned on without any change in the distribution. 
        \item If $v\in D_i$, then the projection operation fixes the variable value $x_{v^{\#}}\equiv x_{v}$. Its effect can hence be removed from the terms being conditioned in $g_i$.
    \end{enumerate}
    For both cases, the projected distribution $\Proj(g_i)$ is consistent with the division of free and fixed vertices in $G[D_i]$; hence it obeys the nested Markov property with respect to $G[D_i]$.
    Note that $\NEST(G)$ is defined as that for every probability distribution $q\in\NEST(G)$, it can be factorized as
\begin{equation}
    q=\prod_{i=1}^{t}h_i(x_{D_i}|x_{\mathrm{pa}_{\VVSplit(G)}(D_i)\backslash D_i}),
\end{equation}
and each of $h_i$ obeys the nested Markov property with respect to the subgraph $G[D_i]$. Comparing the factorization of $\Proj(p)$ with respect to this form, we see that $\Proj(p)$ admits a factorization with respect to the districts of $G$ that is in accordance with Eq.~\eqref{eq:districtdecomp}.

\item Finally, by construction, we naturally have that every childless vertex in $G$ remains a childless vertex in $\VVSplit(G)$. This shows that every variable associated with a childless vertex in $G$, which can be marginalized over in Eq.~\eqref{eq:vermamarginal} of Definition~\ref{def:nested}, can also be marginalized over in $\VVSplit(G)$.
\end{enumerate}

The above discussion shows that if $p\in\IND(\VVSplit(G))$, then $\Proj(p)$ obeys the nested Markov property with respect to $G$. Putting everything together, we have $\VVS(G)=\Proj\bigl(\IND(\VVSplit(G))\bigr)\subseteq\NEST(G)$.
\end{proof}

\section{$e$-separation}\label{app:esep}
In the literature, building on the $d$-separation criteria, which provide a graphical approach to specifying the equality constraints of conditional independences, the $e$-separation (extended $d$-separation) criterion has been developed as an approach to finding inequality constraints among observable variables in a Bayesian network~\cite{evans2012graphical,finkelstein2021entropic}. Specifically, consider a causal DAG $G$ with vertex set $V$ and edge set $\mathcal{E}$, and $A,B,C$, and $D$ are disjoint sets of vertices in $V$ representing observed variables. We denote the causal DAG after deletion of $D$, $G_{V\backslash D}$, as the DAG obtained after removing every vertex in $D$ and all of its incoming and outgoing edges.
Then, we say $A$ and $B$ are $e$-separated given $C$ after the deletion of $D$ in $G$, if $A$ and $B$ are $d$-separated by $C$ in $G_{V\backslash D}$. 

To have a non-trivial $e$-separation among observed variables, namely, that one needs to delete a non-empty vertex set $D$ to specify a $d$-separation in the new graph, the deleted vertices must have both incoming and outgoing edges. That is, they belong to $V_{\mathrm{split}}$ in Definition~\ref{def:SWIG}. In Corollary 4.3 of \citet{evans2012graphical}, it has been shown that if there exists a non-trivial $e$-separation and that all the observed variables represented by vertices in $D$ are discrete, then it leads to a non-trivial inequality constraint among observed variables. Theorem 5 of \citet{finkelstein2021entropic} further provides concrete constructions of entropic inequality constraints. This principle for finding inequality constraints was originally developed in the discussion of Bayesian networks. 

Here, we generalize the proof of \citet{evans2012graphical} and show that the $e$-separation criterion remains valid when the latent variables are allowed to be physical states in a general GPT. Therefore, the inequality constraints induced from $e$-separation are inherent to the causal DAG. We simply need to prove the validity of the following theorem within a general GPT, which is the essence and the only part requiring a generalization to GPTs in the entire proof.

\begin{theorem}[Theorem 4 in \citet{finkelstein2021entropic}; modified from Theorem 4.2 in \citet{evans2012graphical}]\label{thm:e-sep}
    Given a causal DAG $G(V,W,\mathcal{E})$, where $A,B,C,D\subseteq V$ are disjoint sets of vertices in $G$ representing observed variables. Suppose $A$ and $B$ are $e$-separated given $C$ after deletion of $D$ in $G$, and that no variable in $C$ is a descendant of any variable in $D$. Then, given a GPT $\mathcal{G}$, for every valid distribution $p$ over $A,B,C,D$ in $G$, there exists a valid distribution $p^*$ over $A,B,C,D,D^\#$ such that 
    \begin{equation}\label{eq:esep1}
    \begin{split}
        &p(A{=}a,B{=}b,D{=}d|C{=}c) \\ =&p^*(A{=}a,B{=}b,D{=}d|C{=}c,D^\#{=}d),
    \end{split}
    \end{equation}
    with $D^\#$ a random variable of the same alphabet as $D$ and $A\independent B|C,D^\#$ in distribution $p^*$. If furthermore no variable in $A$ is a descendant of any in $D$, then there exists a distribution $p^*$ over $A,B,C,D,D^\#$ such that
    \begin{equation}\label{eq:esep2}
    \begin{split}
        &p(B{=}b,D{=}d|A{=}a,C{=}c) \\ = &p^*(B{=}b,D{=}d|A{=}a,C{=}c,D^\#{=}d),
    \end{split}
    \end{equation}
    with $A\independent B|C,D^\#$ in distribution $p^*$.
\end{theorem}

\begin{proof}
    Consider the visible-vertex-split graph of $G$, $\VVSplit(G)\equiv G'(V\cup V^{\#},W,\mathcal{E})$, with $V_{\mathrm{split}}\subseteq V$ the subset of vertices that are split. As observed for non-trivial $e$-separation, we assume $D\subseteq V_{\mathrm{split}}$ without loss of generality. Within GPT $\mathcal{G}$, suppose that the functioning of an effect on an observed variable $X_v$ in terms of its parental vertices remains the same for every $v\in V$. 
    We can hence assign the same marginal distribution over the observed variables in the subgraph of $G'$, which does not contain $D^{\#}$ and its descendant vertices, as the marginal distribution over the observed variables in the subgraph of $G$, which does not contain descendant vertices of $D$. By construction and rules of GPTs, the assigned distribution is valid within $\mathcal{G}$. For simplicity, we denote the subgraphs as $G[V\backslash \mathrm{de(D)}]$ and $G'[V\backslash (D^{\#}\cup\mathrm{de(D^{\#})})]$.

    Now consider the set of distributions over $V\cup V^{\#}$ compatible with the above distribution over the subgraph. Since vertices in $D^{\#}$ are root vertices in $G'$ and represent observed variables, the variables that they represent can take any probability distribution. Additionally, these variables serve as the classical control of an effect in $\mathcal{G}$. Upon the projection operation of $D^{\#}{=}D{=}d$, the distribution over the observed variables in $G'[V\backslash (D^{\#}\cup\mathrm{de(D^{\#})})]$ is not affected. This is because that by construction, the subgraph's composing observed variables are independent of the variables represented by $D^{\#}$ and their descendants, which can be told from $d$-separation. Furthermore, since the functioning of an effect is the same across $G$ and $G'$, the variables or GPT states of the descendant vertices of $D^{\#}$ in $G'$ will have the same effect as their counterpart variables or states of the descendant vertices of $D$ in $G$. Therefore, upon the projection operation, we obtain a valid distribution within $\mathcal{G}$ over the observed variables in $G'$, in which the marginal distribution over the observed variables represented by $V$ is the same as the distribution $p$ over the observed variables in $G$.

    Based on the above construction, since no variable in $C$ is a descendant of any variable in $D$, the distribution over the variables in $C$ is not affected by the values taken by the variables represented by $D^{\#}$ in $G'$. Eq.~\eqref{eq:esep1} hence holds upon the projection operation of $D^{\#}{=}D{=}d$. Because $A$ and $B$ are $e$-separated given $C$ after deletion of $D$ in $G$, by construction of $\VVSplit(G)$, $A$ and $B$ are $d$-separated given $C$ and $D^{\#}$ in $G'$, and thus $A\independent B|C,D^\#$ in distribution $p^*$. For the same reason, we have Eq.~\eqref{eq:esep1} and $A\independent B|C,D^\#$ in distribution $p^*$ if furthermore no variable in $A$ is a descendant of any in $D$.
\end{proof}

Note that Theorem~\ref{thm:e-sep} naturally holds when considering $p\in \IND(G)$ and $p'\in\IND(\VVS(G))$, which can be easily shown using the $d$-separation criterion. 

Notably, not only does Theorem~\ref{thm:e-sep} itself hold for the visible-vertex-split model, but the entropic inequality constraints of Ref.~\cite{finkelstein2021entropic}, derived there as \emph{consequences} of Theorem~\ref{thm:e-sep}, are \emph{also} inviolable by distributions in the visible-vertex-split model. To see this, one need only inspect that the derivation by~\citet{finkelstein2021entropic} of their entropic inequalities, noting that the entropic arguments concern only relations between distributions over variables corresponding to \emph{visible} vertices across $G$ and $G'$. Since the derivation never needs to refer to entropies pertaining to latent variables, it follows that the derived entropic inequalities hold for the entire visible-vertex-split model, and hence for GPTs as well.

\section{Proof of Proposition~\ref{prop:MINF_saturated}}\label{app:prop_one_latent}
  To prove this proposition, it suffices to prove that ${\GPT(\MaxInt(G))=\IND(\MaxInt(G))=\BW(\MaxInt(G))}$. Since $G$ has only one latent variable, the maximal interruption graph $\MaxInt(G)$ corresponds to a standard Bell test, where the number of nonlocal parties is equal to the number of child vertices of the latent variable in $G$. In $\MaxInt(G)$, ${\GPT(\MaxInt(G))\subseteq\IND(\MaxInt(G))}$ holds, as the non-trivial requirements in $\IND(\MaxInt(G))$ correspond to no-signalling conditions, which should be satisfied by any valid GPT. Also, ${\IND(\MaxInt(G))=\BW(\MaxInt(G))}$, where $\BW$ corresponds to the BoxWorld, which is a valid GPT under assumptions in Box~\ref{box:GPTAssumption}~\cite{popescu1994quantum,barrett2007information,gross2010all}. Since ${\BW(\MaxInt(G))\subseteq\GPT(\MaxInt(G))}$, we have that $\GPT(G)=\MI(G)$.

\section{Definition of causal copy isomorphism}\label{app:defcopyisomorphism}

\begin{figure*}[hbtp!]
\centering
 \includegraphics[width=\textwidth]{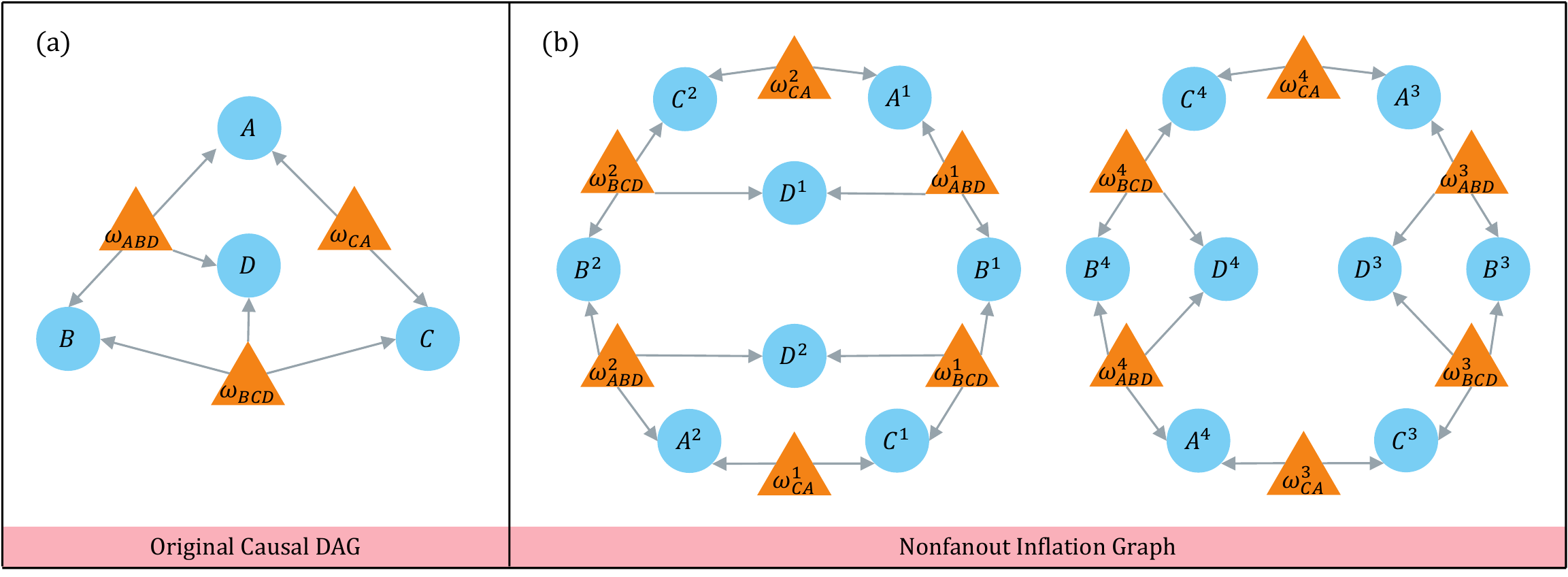}
\caption{Significance of causal copy isomorphisms. (a) An example causal DAG. It involves four observed variables, $A,B,C$, and $D$, which are denoted by blue circles, and three latent variables, $\omega_{ABD},\omega_{BCD}$, and $\omega_{CA}$, which are denoted by orange triangles. (b) A nonfanout inflation of the original causal DAG. Every vertex has a name and a copy index, which indicate how it is related to the vertices in the original causal DAG.}
\label{fig:TriBDinflation}
\end{figure*}

\begin{definition}[\textbf{Causal copy isomorphism}]
Given a causal DAG $G$, consider a nonfanout inflation graph of it, $G'(V',W',\mathcal{E}')$, where $V'$ is the set of vertices representing observed variables, $W'$ is the set of vertices representing latent variables, and $\mathcal{E}'$ is the edge set, and where each vertex has both a name and a copy index. Let $V^{\prime 1}$ and $V^{\prime 2}$ be two ordered subsets of $V'$ with the same cardinality. These two ordered sets are \textbf{causally copy isomorphic} if and only if there exists a bijective map $\phi$ relating the set of \emph{all} the vertices of $G'$ to itself, such that:

\begin{enumerate}[nosep]
\item $\phi(V^{\prime 1})=\phi(V^{\prime 2})$.

\item Let $\mathcal{E}^{\prime 1}$ be the (orderless) set of all edges in $\mathcal{E}'$ that terminate at either an element of $V^{\prime 1}$ or an ancestor of $V^{\prime 1}$ (i.e., $\mathcal{E}^{\prime 1}$ is the set of all edges pertinent to the causal ancestry of $V^{\prime 1}$), and let $\mathcal{E}^{\prime 2}$ be the corresponding collection of edges pertinent to the causal ancestry of $V^{\prime 2}$. Then, the bijective map on vertices must also relate these two edge sets; that is, $\phi(\mathcal{E}^{\prime 1})=\phi(\mathcal{E}^{\prime 2})$.

\item $\phi$ preserves names and only ever modifies a vertex's copy index.
\end{enumerate}
\end{definition}

Informally, $V^{\prime 1}$ and  $V^{\prime 2}$ are causally copy-isomorphic if and only if their ancestral subgraphs coincide under the deflation map;
that is, once we ignore differences between vertices that differ only by their copy index, the two ancestral subgraphs become identical.

As an intuitive example, consider the causal DAG in Figure~\ref{fig:TriBDinflation}(a), which involves four observable variables and three latent variables. We depict one of its nonfanout inflation graphs in Figure~\ref{fig:TriBDinflation}(b). First, note that this inflation graph has disconnected components, which can happen. As per Definition~\ref{Def:Injectabl}, some injectable sets are $\{A^1,C^2,D^1\}$ and $\{A^1,B^1\}$ and $\{B^3,D^3\}$. The existence of causal copy isomorphisms implies that $p(A^1,B^1,C^1,A^2,B^2,C^2)=p(A^2,B^2,C^2,A^1,B^1,C^1)$ (per a causal copy \emph{automorphism}) as well as $p(A^1,B^1,C^1,A^2,B^2,C^2)=p(A^3,B^3,C^3,A^4,B^4,C^4)$.

\section{A minimal causal structure separating causal models}\label{app:GNYI}
As is extensively explored in Bell nonlocality, we are interested in finding causal structures where different physical theories pose different constraints to the correlations among observed variables. In this section, we find a new DAG such that
\begin{equation}
    \C(G) \subsetneq \GPT(G)=\VVS(G) \subsetneq \NEST(G).
\end{equation}

Consider the DAG $G$ shown in Figure~\ref{fig:GYNDAG}(a), with $X,A,B,C\in\{0,1\}$. In this causal structure, a latent variable $\Lambda$ simultaneously affects observables $A,B$, and $C$. An exogenous variable $X$ directly affects $A$, which is mediated by $B$ to be carried to $C$. The nested Markov model of this causal structure trivially comprises all the probability distributions.

\begin{figure}[b!]
 \centering \includegraphics[width=\columnwidth]{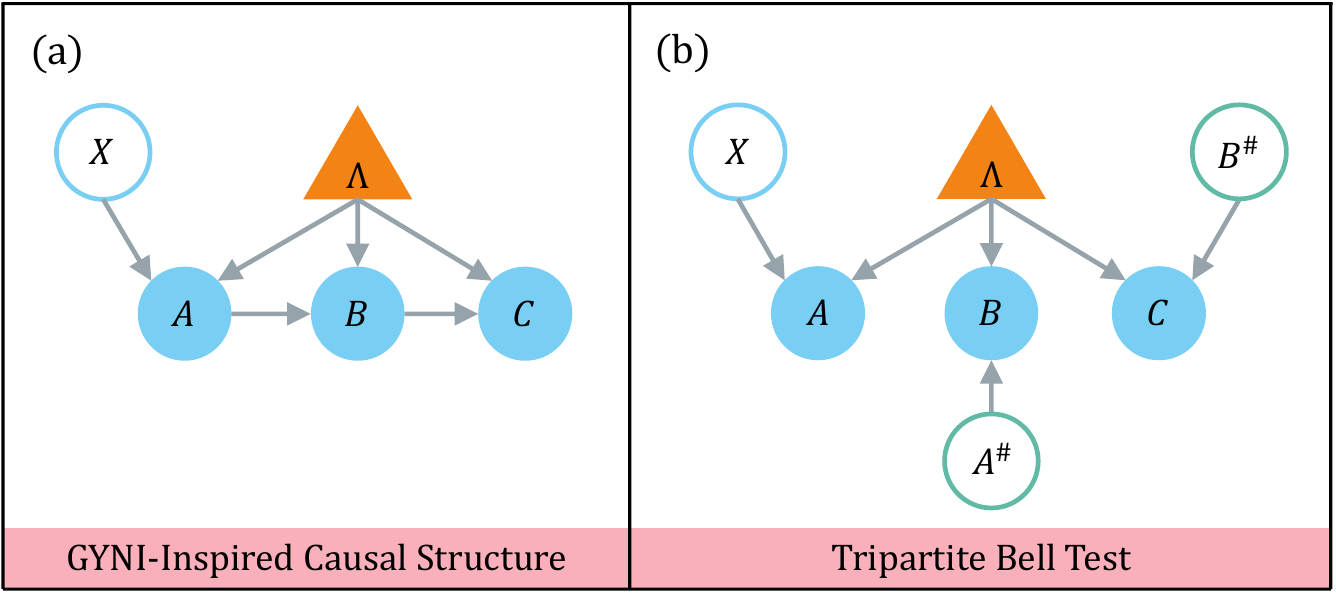}
\caption{A causal structure with $\C(G)\subsetneq\GPT(G)=\VVS(G)\subsetneq\NEST(G)$. (a) The causal DAG of the structure. The circles represent observable variables, including $X,A,B$, and $C$, and the triangle represents the latent variable $\Lambda$. (b) The visible-vertex-split graph of (a), which is the tripartite Bell test. The green hollow circles are added vertices, and a data post-selection of $A^{\#}{=}A$ and $B^{\#}{=}B$ projects the structure back to the original scenario.}
\label{fig:GYNDAG}
\end{figure}

From Proposition~\ref{prop:already_seminetwork}, we have $\NEST(G)=\IND(G)$. At the same time, as $G$ contains the instrumental scenario as a sub-structure, where the no-signalling condition poses the non-trivial condition in Eq.~\eqref{eq:instrument} on the observational data, $\GPT(G)$ is a strict subset of $\NEST(G)$. Additionally, we have $\GPT(G)=\VVS(G)$ for this DAG per Proposition~\ref{prop:MINF_saturated}.

To show the strict inclusion of $\C(G)\subsetneq\GPT(G)$, we look into the visible-vertex-split graph of $G$. As shown in Figure~\ref{fig:GYNDAG}(b), the high-dimensional structure represents the tripartite standard Bell test with input variables $X,A^{\#}$, and $B^{\#}$, which set random measurements on a tripartite physical state $\Lambda$ with measurement results $A,B$, and $C$, respectively. In the tripartite Bell test, there exist no-signalling nonlocal statistics---which represent valid physical states in the BoxWorld---that cannot be generated by any classical or even quantum operations~\cite{almeida2010guess}. One such distribution is given by
\begin{equation}
  p(a,b,c|x,a^{\#},b^{\#})=P_1(a,b,c|x,a^{\#},b^{\#})\equiv\frac{1}{3}f(a,b,c,x,a^{\#},b^{\#}), 
\end{equation}
with
\begin{equation}
\begin{split}
  f(a,b,c,x,a^{\#},b^{\#})=&(1\oplus b\oplus x\oplus a^{\#}\oplus xa^{\#})(1\oplus c\oplus b^{\#}) \\
  &\oplus a(1\oplus a^{\#}\oplus ca^{\#}\oplus b(c\oplus b^{\#})),
\end{split}
\end{equation}
where all the variables take values in $\{0,1\}$. Intuitively, the distribution can be understood as an optimal strategy to play a nonlocal game called ``guess your neighbour's input'' (GYNI)~\cite{almeida2010guess}, where three nonlocal players win the game if and only if $a{=}a^{\#}, b{=}b^{\#}$, and $c{=}x$, namely, each player's output makes a correct guess of their neighbour's input. Except $c{=}x$, the winning condition is inherently similar to the projection operation in defining the visible-vertex-split model. Consequently, we obtain a set of valid statistics in $\GPT(G)$:
\begin{equation}
  p(a,b,c|x)=P_1(a,b,c|x,a^{\#}{=}a,b^{\#}{=}b),
\end{equation}
with
\begin{equation}\label{eq:GYNIproj1}
\begin{split}
    &p(a,b,c|x)
    \\
    =&
    \begin{cases}
        \frac{1}{3}, & \mbox{if } (a,b,c)=(0,0,0), (1,0,1), (1,1,0) \mbox{ when } x{=}0, \\
        \frac{1}{3}, & \mbox{if } (a,b,c)=(0,1,1),(1,0,1),(1,1,0) \mbox{ when } x{=}1, \\
        0, & \mbox{otherwise}.
    \end{cases}
\end{split}
\end{equation}
We prove that Eq.~\eqref{eq:GYNIproj1} does not admit a causal explanation in a classical Bayesian network, leading to the strict separation of $\C(G)\subsetneq\GPT(G)$. To prove this result, first note that the set of probability distributions that can be generated from the causal DAG $G$ of Figure~\ref{fig:GYNDAG}(a) within a GPT is convex, as there is only one latent variable. In particular, the set of distributions $\C(G)$ within a Bayesian network is a convex polytope. Therefore, we can list all extremal distributions of $\C(G)$. 

For this purpose, we can resort to the visible vertex split graph in Figure~\ref{fig:GYNDAG}(a), determine all the extremal distributions in $\C(\VVSplit(G))$, and utilize the projection in Definition~\ref{def:projection}. When all the observed variables are binary, the set of probability distributions of $\C(\VVSplit(G))$ is spanned by the following extremal points:
\begin{equation}
    p(a,b,c|x,a^{\#},b^{\#})=\begin{cases}
    1,& \begin{smallmatrix}a=\alpha x\oplus\beta,\\b=\gamma a^{\#}\oplus\delta,\\c=\theta b^{\#}+\eta,\end{smallmatrix}\\
    0,& \text{otherwise},
    \end{cases}
\end{equation}
where $\alpha,\beta,\gamma,\delta,\theta,\eta \in \{0,1\}$. In total, there are $2^6=64$ extreme points in $\C(H(G))$. Next, we take the projection and obtain the extremal points in $\C(G)$ by taking $A^{\#}=A,B^{\#}=B$. Note that the projection results in some degeneracy. After the projection, the values of $A=a,B=b$, and $C=c$ are determined by
\begin{equation}
\begin{split}
    a&=\alpha x\oplus\beta,\\
    b&=\gamma a\oplus\delta=\gamma (\alpha x\oplus\beta)\oplus\delta,\\
    c&=\theta b+\eta=\theta (\gamma (\alpha x\oplus\beta)\oplus\delta)+\eta.
\end{split}
\end{equation}
Afterwards, we can set up a linear programming optimization to test whether Eq.~\eqref{eq:GYNIproj1} can be written as a convex combination of these extremal points. This is impossible, hence proving that the probability distribution in Eq.~\eqref{eq:GYNIproj1} is outside $\C(G)$.

\newpage
\nocite{apsrev42Control}
\bibliographystyle{apsrev4-2}
\bibliography{bibCausal}

\begin{acknowledgments}
We acknowledge Lorenzo Giannell, Yuwei Zhu, Yuan Liu, Thomas Richardson, and Stefano Pironio for the insightful discussions.
XZ acknowledges support by the Hong Kong Research Grant Council (RGC) through grant number R7035-21, SRFS2021-7S02, No. 27300823, and N\_HKU718/23, HKU Seed Fund for Basic Research for New Staff via Project 2201100596, Guangdong Natural Science Fund via Project 2023A1515012185, National Natural Science Foundation of China (NSFC) via Project No. 12305030 and No. 12347104, and Guangdong Provincial Quantum Science Strategic Initiative GDZX2200001. This research was supported in part by Perimeter Institute for Theoretical Physics. Research at Perimeter Institute is supported in part by the Government of Canada through the Department of Innovation, Science and Economic Development and by the Province of Ontario through the Ministry of Colleges and Universities.
\end{acknowledgments}

\end{document}